\documentclass[format=acmsmall, review=false, screen=true]{acmart}

\usepackage{booktabs} 

\usepackage{algorithmic}
\usepackage{amsmath}
\usepackage{makecell}
\usepackage{caption}

\usepackage{lineno,hyperref}
\modulolinenumbers[5]
\usepackage{amssymb}
\usepackage{enumerate}
 \usepackage{epsfig}
\usepackage{amsfonts}
\usepackage{setspace}
\usepackage{fancyhdr}
\usepackage{amsmath}
\usepackage{amssymb}                                                                                                                                                                                                                                                                                                                                                                                                                                                                                                                                                                                                                          
\usepackage{graphicx}
\usepackage{xcolor}
\usepackage[misc]{ifsym}

\DeclareMathOperator*{\argmax}{argmax}
 
\newcommand\mpage[1]{\gape{$\vcenter{\hbox{#1}}$}}

 \newcommand{\tabincell}[2]{\begin{tabular}{@{}#1@{}}#2\end{tabular}}

\usepackage[ruled]{algorithm2e} 

\SetAlFnt{\small}
\SetAlCapFnt{\small}
\SetAlCapNameFnt{\small}
\SetAlCapHSkip{0pt}
\IncMargin{-\parindent}

\acmJournal{TKDD}
\acmVolume{9}
\acmNumber{4}
\acmArticle{39}
\acmYear{2010}
\acmMonth{3}
\copyrightyear{2009}

\setcopyright{acmlicensed}

\acmDOI{0000001.0000001}

\received{February 2007}
\received[revised]{March 2009}
\received[accepted]{June 2009}

\begin{document}
\title[Time-sync Video Tag Extraction Using SAG]{Time-sync Video Tag Extraction Using Semantic Association Graph}

\author{Wenmain Yang}
\orcid{0000-0001-8493-4449}
\affiliation{%
  \institution{Shanghai JiaoTong University; University of Macau}
  \country{China}}
\email{sdq11111@sjtu.edu.cn}
\author{Kun Wang}
\affiliation{%
  \institution{University of California, Los Angeles}
  \postcode{90095}
  \country{America}}
\email{wangk@ucla.edu}
\author{Na Ruan}
\author{Wenyuan Gao}
\affiliation{%
  \institution{Shanghai JiaoTong University}
  \streetaddress{No. 800, Dongchuan Road}
  \city{Shanghai}
  \postcode{200240}
  \country{China}}
\author{Weijia Jia*}
\affiliation{%
  \institution{University of Macau; Shanghai JiaoTong University}
\streetaddress{Corresponding Author}
  \city{Macau}
\postcode{999078}
  \country{China}
}
\email{jiawj@um.edu.mo}
\author{Wei Zhao}
\affiliation{%
 \institution{American University of Sharjah}
 \city{Sharjah}
 \country{United Arab Emirates}}
\email{weizhao@umac.mo}
\author{Nan Liu}
\author{Yunyong Zhang}
\affiliation{%
  \institution{China Unicom Research Institute}
  \streetaddress{Bldg.2, No,1 Beihuan East Road, Economic-Technological Development Area}
  \city{Beijing}
\postcode{100176}
  \country{China}
}

\begin{abstract}
Time-sync comments reveal a new way of extracting the online video tags. However, such time-sync comments have lots of noises due to users' diverse comments, introducing great challenges for accurate and fast video tag extractions. In this paper, we propose an unsupervised video tag extraction algorithm named Semantic Weight-Inverse Document Frequency (SW-IDF). Specifically, we first generate corresponding semantic association graph (SAG) using semantic similarities and timestamps of the time-sync comments. Second, we propose two graph cluster algorithms, i.e., dialogue-based algorithm and topic center-based algorithm, to deal with the videos with different density of comments. Third, we design a graph iteration algorithm to assign the weight to each comment based on the degrees of the clustered subgraphs, which can differentiate the meaningful comments from the noises. Finally, we gain the weight of each word by combining Semantic Weight (SW) and Inverse Document Frequency (IDF). In this way, the video tags are extracted automatically in an unsupervised way. Extensive experiments have shown that SW-IDF (dialogue-based algorithm) achieves 0.4210 F1-score and 0.4932 MAP (Mean Average Precision) in high-density comments, 0.4267 F1-score and 0.3623 MAP in low-density comments; while SW-IDF (topic center-based algorithm) achieves 0.4444 F1-score and 0.5122 MAP in high-density comments, 0.4207 F1-score and 0.3522 MAP in low-density comments. It has a better performance than the state-of-the-art unsupervised algorithms in both F1-score and MAP.
\end{abstract}

%
%

\ccsdesc[500]{Information systems~Data mining}
\ccsdesc[500]{Information systems~Video search}

%
%

\keywords{multimedia retrieval, video tagging, crowdsourced time-sync comments, semantic association graph, keywords extraction}

\maketitle

\renewcommand{\shortauthors}{Yang, W. et al.}

\section{Introduction}

Recently, watching online videos of news and amusement have become mainstream entertainment during people's leisure time. The booming of online video-sharing websites raises significant challenges in effective management and retrieval of videos. To address that, many text retrieval based automatic video tagging techniques have been proposed \cite{siersdorfer2009automatic,raamkumar2017using,ramaboa2018keyword}. However, these techniques can only provide video-level tags \cite{wu2014crowdsourced}. The problem is that even if these generated tags can perfectly summarize the video content, users have no idea how these tags are associated with the video playback time. If videos are associated with time-sync tags, users can preview the content with both thumbnails and text along the timeline, and this textual information can further enhance users' search experience. Although there are many video content analysis algorithms that can generate video tags with timestamps \cite{hussein2017v,chen2017video}, their time complexities are too high for large-scale video retrieval. Fortunately, a new type of review data, i.e., time-sync comments (TSCs) appear on video websites like Youku (www.youku.com), AcFun (www.acfun.tv) and BiliBili (www.bilibili.com) in China, and NicoNico (www.nicovideo.jp) in Japan.

In this paper, we focus on extracting time-sync video tags from TSCs efficiently, which can enhance users' search experience. When watching a video, many people are willing to share their feelings and exchange ideas with others. TSC is such a new form of real-time and interactive crowdsourced comments \cite{wang2017crowdsourcing,wang2016toward,gu2017reliable,hyung2017utilizing}. TSCs are displayed as streams of moving subtitles overlaid on the video screen, and convey information involving the content of current video frame, feelings of users or replies to other TSCs. In TSC-enabled online video platforms, users can make their comments synchronized to a video's playback time. That is, once a user posts a TSC, it will be synchronized to the associated video time and immediately displayed onto the video. All viewers (including the writer) of the video can see the TSCs when they watch around the associated video time. Moreover, each TSC has a timestamp to record the corresponding video time when posted. Therefore, compared with traditional video reviews, TSCs are much easier to obtain the local tags with timestamp rather than video-level tags. Moreover, the TSCs are more personalized than traditional reviews, therefore the tags generated by TSCs can better reflect the user's perspective. The users can thereby get high-quality retrieval results when they search for videos with these tags \cite{wu2014crowdsourced}.

Recently, some methods have been proposed to generate temporal tags or labels based on TSCs. Wu et al. \cite{wu2014crowdsourced} use statistics and topic model to build Temporal and Personalized Topic Modeling (TPTM) to generate temporal tags. However, their approach is based on the Latent Dirichlet Allocation (LDA) model \cite{blei2003latent}, which has poor performance when dealing with short and noisy text like TSC \cite{yan2013biterm}.
Lv et al. \cite{lv2016reading} propose a Temporal Deep Structured Semantic Model (T-DSSM) to generate video labels in a supervised way. However, their approach does not consider the semantic association between TSCs, so that some of the video content-independent noises cannot be processed. In summary, TSCs have some features distinguished from the common comments \cite{yang2017crowdsourced,liao2018tscset}, which make the above methods not very effective in the TSCs:
\emph{(1) Semantic relevance.} Abundant video semantic information is contained that describes both local and global video contents by selecting the time interval of the timestamp.
\emph{(2) Real-time.} TSC is synchronous to the real-time content of the videos. Users may produce different topics when it comes to the same video contents.
\emph{(3) Herding effects.} Herding effects are common in TSCs \cite{he2016predicting,yu2015discovering}. That means, latter TSCs may depend on the former ones and have a semantic association with the preceding ones.
\emph{(4) Noise.} Some video content-independent comments and internet slang are included in TSCs, which makes trouble for tag extraction.
Due to the above features of TSCs, how to deal with the herding effects, distinguishing the importance of each TSC and consequently identify high-impact TSCs and noises are the major challenges for extracting video tags from TSCs. 

To make full use of the features of TSC and tackle the above challenges, we propose a graph-based algorithm named Semantic Weight-Inverse Document Frequency (SW-IDF) to generate time-sync video tags automatically. More precisely, we design to reduce the impact of noises by clustering the semantic similar and time-related TSCs and identify high-impact TSCs by their semantic relationships. Intuitively, TSCs including video tags are usually within hot topics and impact on the trend of their follow-up TSCs. On the contrary, the noises usually neither have similar semantic relationships with other TSCs over a period nor influence other TSCs \cite{yang2017crowdsourced}. Moreover, we find that the density of TSCs (number of TSCs per unit time) affects how users communicate. When the density is low (the TSC in a period is sparse), the user can more clearly distinguish the content of each nearby TSC, and therefore is more likely for the user to reply to a specific TSC when posting the new one. Conversely, when the density is high (the TSC in a period is dense), the user cannot clearly distinguish the content of each TSC, but only roughly distinguish the topic of these TSCs. Therefore, the user is more likely to reply to the entire topic instead of a specific TSC.
Specifically, in the SW-IDF algorithm, we first treat the TSCs as vertices, generating the semantic association graph (SAG) based on semantic similarities and timestamps of TSCs.
Then, we intend to cluster TSCs into different topics. For the videos with low-density TSCs, we propose a dialogue-based clustering algorithm, which is inspired by community detection theory \cite{huang2017overlapping,fortunato2010community,lancichinetti2009community}. For the videos with high-density TSCs, we propose a topic center-based cluster algorithm, which is a novel hierarchical agglomerative clustering \cite{pandove2018systematic,murtagh2014ward,murtagh2012algorithms}. These two cluster algorithms can identify the topic of each TSC and distinguish the popularity of each topic in any case.
In the clustered subgraph, the in-degrees of each TSC express its affecting TSCs, while the out-degrees express its affected TSCs. Therefore, we design a graph iteration algorithm to assign the weight of each TSC by its degrees so that we can differentiate the meaningful TSCs from noises. Moreover, similar to TF-IDF algorithm, we gain the weight of each word by combining Semantic Weight (SW) and Inverse Document Frequency (IDF) and the video tags are extracted automatically.

Particularly, this paper is an extended version of \cite{yang2017crowdsourced}. In this extended version, we propose a novel topic center-based cluster algorithm at first, which is more suitable for high-density TSCs. Then, we provide a greedy optimization for the topic center-based algorithm and prove this optimization will not delete any valid case. Finally, we add more experiments to verify the effectiveness of the algorithms.
The main contributions of our paper are as follows:
\begin{enumerate}[1)]
\item We propose a novel graph-based Semantic Weight-Inverse Document Frequency (SW-IDF) algorithm, which can extract video tags in an unsupervised way by mining TSCs.

\item We design two graph clustering algorithms based on the density of the TSCs, i.e., dialogue-based clustering algorithm and topic center-based cluster algorithm, to cluster in the semantic association graph (SAG). These algorithms take the features of TSCs into account and effectively reduce the impact of noises.

\item We evaluate our proposed algorithms with real-world datasets on mainstream video-sharing websites and compare results with classical keyword extraction methods. The results show that SW-IDF outperforms baselines in both precision and recall of video tag extraction.
\end{enumerate}

In the rest of the paper, we introduce the related work in Section \ref{sec:2}, and then formally propose our algorithm in Section \ref{sec:3}. In Section \ref{sec:4}, we verify the effectiveness of our algorithm with experimental results. Conclusions of this work are presented in Section \ref{sec:5}.

\section{Related Work}

\label{sec:2}
In this section, we introduce the related work from four aspects.

\subsection{Analysis of time-sync video comments}
Time-Sync Comments (TSCs) provide a new source of information regarding the video and have received growing research interests. Wu et al. \cite{wu2014crowdsourced} first introduce TSCs and propose a Temporal and Personalized Topic Modeling (TPTM) to generate temporal tags. However, their approach is based on the Latent Dirichlet Allocation (LDA) model \cite{blei2003latent}, which has poor performance when dealing with short text like TSC \cite{yan2013biterm}. To describe the video more specifically, Xu and Zhang \cite{xu2017bridging} extract representative TSCs based on a temporal summarization model. Their methods need the pre-extracted keywords in the TSCs, so our algorithm can improve the effectiveness of them. There are also some other applications based on TSCs. Lv et al. \cite{lv2016reading} propose a Temporal Deep Structured Semantic Model (T-DSSM) to represent comments as semantic vectors and recognize video highlights by semantic vectors in a supervised way. They are the first to analyze the TSC using the neural network. Then, Chen et al. \cite{chen2017personalized} propose the neural network based collaborative filtering to recommend the personalized keyframe from TSCs. However, both the models of \cite{lv2016reading} and \cite{chen2017personalized} rely on a large amount of human-labeled video segments or predefined emotional tags to train, which limits its applicability to more general scenarios. In this paper, we design a novel graph-based algorithm according to the features of TSC to efficiently and accurately extract keywords automatically in an unsupervised way.

\subsection{Tag/keyword extraction} Keyword extraction is a classical problem in the field of information retrieval. At present, mainly three categories unsupervised keyword extraction methods are available. The first one is based on word frequency statistics, where TF-IDF is the most commonly used and well-known method. However, this kind of methods only consider the frequency of words and ignore the semantics, which may generate keywords that are not related to video content. The second kind of methods depends on the co-occurrence of words, such as textrank \cite{mihalcea2004textrank}, which is a graph-based ranking model. Similar to the first one, this kind of methods does not consider semantics either, so it cannot solve the noise well. And the last one is according to the topic model. It brings document-topic and topic-word distribution together by simulating document generation process. Blei et al. \cite{blei2003latent} propose the Latent Dirichlet Allocation(LDA) model, the most representative model. To better deal with short text situation, Yan et al. \cite{yan2013biterm} propose the Bi-term Topic Model (BTM), which models the generation of word co-occurrence patterns (i.e., bi-terms) in the whole corpus directly. Yin and Wang \cite{yin2014dirichlet,yin2016model} propose the Gibbs Sampling algorithm for the Dirichlet multinomial mixture model for short text clustering and keyword extraction. Although the topic model-based approaches consider the semantics, their basic hypothesis is that the generation of each word is independent and identically distributed. However, some TSCs are generated by herding effects, which does not satisfy the assumptions. Compared with the methods above, our algorithms are well-designed to identify noises by analyzing the semantic relationship between TSCs.

\subsection{Semantic similarity} Semantic similarity calculation is an essential issue of natural language processing, which is widely used in text classification \cite{wang2016world}, fuzzy retrieval \cite{alhabashneh2017fuzzy}, and so on. Generally, there are mainly two kinds of approaches to measuring the similarity of documents. One is based on the similarity of the words in sentences. The representations of this approach are proposed by \cite{kenter2015short} on unsupervised learning and \cite{socher2011parsing,kusner2015word} on supervised learning. Considering that time-sync comments contain a mass of newborn internet slangs, it is difficult to obtain accurate results in this way. The other one is based on the sentence vector. The topic model such as LDA, and embedding model such as word2vec \cite{mikolov2013distributed,levy2015improving} are the representations of this kind of methods. Since the embedding model offers much denser feature representation, embedding based similarity computation is better TSCs than the topic model-based methods. Kenter and De Rijke \cite{kenter2015short} propose a supervised learning method based on external sources of semantic knowledge with word embedding, which considers the weight of the semantic feature. In this paper, we only consider the topics discussed by TSCs while the word order will not change the topics discussed in the TSCs. Therefore, the word order is not important and the sentence2vec \cite{iacobacci2015sensembed,levy2014neural} and deep learning  \cite{he2015multi,mueller2016siamese} based methods are not used in this paper.

\subsection{Graph clustering algorithm} Graph clustering algorithms have attracted much research interest in the past. There are two main theories, i.e., community detection theory and hierarchical agglomerative clustering inspired our work. Community detection theory is first proposed by \cite{newman2004finding} to make natural divisions of network nodes into densely connected subgroups, which brings great inspiration to the graph clustering field. Recently, Ramezani et al. \cite{ramezani2018community} exploit the diffusion information and utilize the conditional random fields to discover the community structures. Li et al. \cite{li2018local} propose a novel local expansion via minimum one norm approach for finding overlapping communities, and provide the theoretical analysis of the local spectral properties. Chakraborty et al. \cite{chakraborty2016permanence} find that the belongingness of nodes in a community is not uniform and design a new vertex-based metric to quantify the degree of belongingness within a community. To reduce the time complexity, Bae et al. \cite{bae2017scalable} propose an algorithm to optimize the map equation, which makes the iterations take less time, and the algorithm converges faster. These above-mentioned community detection theory based graph clustering algorithms provide us with good inspiration for designing dialogue-based clustering algorithms. Besides, hierarchical agglomerative clustering is also a method of graph clustering \cite{pandove2018systematic,murtagh2014ward,murtagh2012algorithms}. Recently, Pang et al. \cite{pang2015unsupervised} propose a topic-restricted similarity diffusion process to efficiently identify real topics from a large number of candidates. Although their method has a good clustering effect, it has a high time complexity and is not suitable for large-scale data. Compared with the aforementioned hierarchical agglomerative clustering algorithms, we proposed a novel topic center-based clustering algorithm have lower time complexity under the condition of ensuring accuracy.

\section{Algorithms}
\label{sec:3}
In this section, we first introduce the construction of Semantic Association Graph (SAG) for TSCs with their semantic similarity in Section \ref{graph}. Then, we propose two graph cluster algorithms, i.e., dialogue-based algorithm and topic center-based algorithm, to cluster the TSCs into subgraphs of different topics in Section \ref{3.2}. Moreover, we propose an out-in degree iterative algorithm to get the weight of each TSC and extract keywords as video tags automatically by combining Semantic Weight (SW) and inverse document frequency (IDF) in Section \ref{3.4}. Finally, we give the complexity analysis in Section \ref{ca}. 

The Notation list is shown in Table \ref{notation}.

\subsection{Preliminaries and Graph Construction}
\label{graph}
\begin{table}[!t]
\centering
\caption{Notations}
\label{notation}
\begin{tabular}{|c|l|}
\hline
\hline
$G$ & Directed graph \\
\hline
$V$ & Set of nodes \\
\hline
$E$ & Set of Edges \\
\hline
$N$ & Number of nodes in $V$ \\
\hline
$M$ & Number of edges in $E$ \\
\hline
$v_{i}$ & $i-th$ node in $V$ \\
\hline
$e_{i}$ & $i-th$ edge in $E$\\
\hline
$t_{i}$ & Timestamp of  node $i$\\
\hline
$v_{i}.S$ &Topic set of node $v_{i}$\\
\hline
$|S|$ & Number of nodes in set $S$\\
\hline
$e_{i}.x$ & The first node of edge $i$\\
\hline
$e_{i}.y$ & The second node of edge $i$\\
\hline
$e_{i}.w$  & The weight of edge $e_{i}$\\
\hline
$\gamma_{t}$ & Attenuation coefficient\\
\hline
$ \rho_{d}$ & Threshold of dialogue bsed intra-cluster density\\
\hline
$ \rho_{c}$ & Threshold of topic based intra-cluster density\\
\hline
$vec_i $ & The embedding vector of TSC $i$\\
\hline
$S.center$ & Topic center vector of the set $S$\\
\hline
$S.st$ & Start time of the set $S$\\
\hline
$S.ct$ & Center time of the set $S$\\
\hline
$ST$ & Universal set of topic sets\\
\hline
$match_{i}$ &A set that matches $S_{i}$\\
\hline
$maxval_{i}$ &Max Affinity value of set $S_{i}$\\
\hline
${\rm Aff}Q$ & A priority queue with set pairs\\
\hline
$Ulist$ & A queue with sets to be updated\\
\hline
$P_{i}$ & Popularity of comment $i$\\
\hline
$K$ & Total number of topics in SAG\\
\hline
$\mathbb{M}_{N \times N}$ & Influence matrix\\
\hline
$I_{i,k}$ & Influence value of comment i after k iterations\\
\hline
$W_{i}$ & Weight of comment $i$\\
\hline
\hline
\end{tabular}
\end{table}

In this section, we construct the semantic association graph and define the attributes in the graph.

Since TSCs appear in chronological order, they can only affect the upcoming TSCs rather than prior TSCs. We use a directed graph to describe the relationships between TSCs and construct the semantic association graph (SAG). 

In SAG, the vertices (nodes) are TSCs and the edges reflect their semantic association in a topic. Let $G$ denote the directed graph, represented by $G =(V,E)$, where $V$ and $E$ are the sets of nodes and edges. Specifically, $V=\{v_{1}, v_{2}, ..., v_{N}\}$, $E=\{e_{1}, e_{2}, ..., e_{M}\}$, where $N$ is the number of nodes in $V$, and $M$ is the number of edges in $E$. For each TSC $i$, it has a timestamp $t_{i}$, denoting the post time in the video, where $t_{v_{1}}<t_{v_{2}}<...<t_{v_{N}}$. Since the TSCs are the short texts \cite{wu2014crowdsourced}, in our algorithm, we assume that each TSC has one exact topic. For vertex $v_{i}$, $v_{i}.S$ is used to describe the set that contain the vertices which have the same topic as $v_{i}$ and $|S|$ is used to express the number of vertices in set $S$. We use the domain to describe the attributes of edges. For edge $e_{k}$, $e_{k}.x$ and $e_{k}.y$ are two vertices that are linked by edge $e_{k}$ where $t_{e_{k}.x} < t_{e_{k}.y}$. The weight of edge $i$ is described as $e_{i}.w$. Besides, $e_{u,v}$ also describes the edge with vertices $u$ and $v$ where $t_{u}<t_{v}$. Next we will provide the definition of edge weights.

As we mentioned in Section \ref{sec:2}, an embedding based method word2vec (more details see Section \ref{exset}) is selected to calculate the semantic similarity between each pair of TSCs. Since we only care about the topic of the TSC, the word order is not important. In this paper, we calculate the mean vector of each word in a TSC as the sentence vector. We set the dimension of each vector as $d$. Therefore, the semantic similarity between TSC $a$ and $b$ is calculated by the cosine angle between vectors:
 \begin{equation}
Sim(a,b) = \frac{\vec a \cdot \vec b}{|\vec a| |\vec b|}.
 \end{equation}

Besides, the greater the timestamp interval between two TSCs, the less likely they are in the same topic. So we use the exponential function to express the decay of TSC associations:
 \begin{equation}
\label{decay}
Delay(a,b) = exp^{-\gamma_{t} \cdot (t_{b}-t_{a})},
 \end{equation}
where $\gamma_{t}$ is a hyperparameter that control the decay speed.

Combining the semantic similarity and the time decay, the weight of edge $i$ that link vertices $u$ and $v$ is defined by
\begin{equation}
e_{i.w}=\begin{cases}
Sim(u,v)\cdot Delay(u,v) &\text{if $t_{u}<t_{v}$}\\
0 &\text{if $t_{u}>t_{v}$}\end{cases}.
\end{equation}

Empirically derived threshold, two TSCs with a negative weight edge are less semantically related (because their angle in the semantic embedding space is greater than $\pi/2$), and negative edge weights are inconvenient to calculate in graph algorithms. Therefore, when $e_{u,v}.w<0$, we set $e_{u,v}.w=0$ and delete this edge.

\begin{figure*}[t]
\centering
\includegraphics[width=1\linewidth]{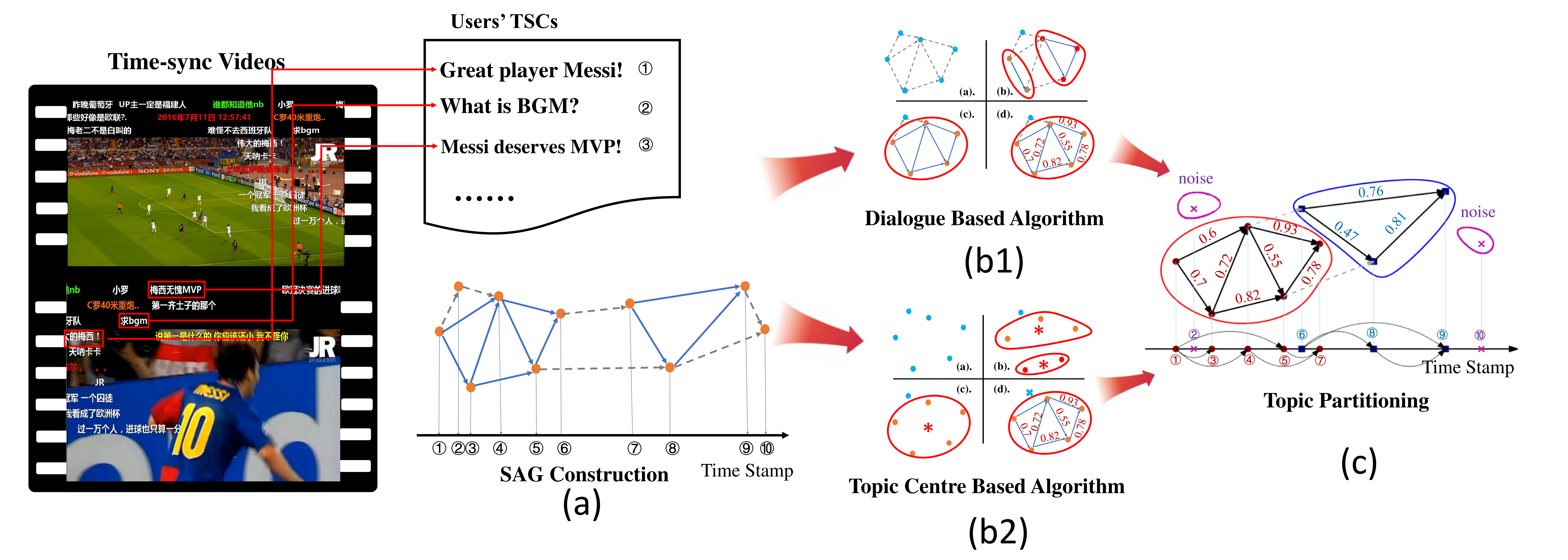}
\caption{An example of SAG Construction}\label{fig:graph}
\end{figure*}
For a more intuitive description, an example of SAG construction is shown in Fig. \ref{fig:graph} (a), which is a UEFA Champions League video. We select 10 TSCs as nodes and construct the SAG. User A made the TSC \textcircled{1} as ``Great player Messi!'' when he saw the goal. Then user B responded with ``Messi deserves MVP!" as the TSC \textcircled{3}. User C makes a TSC ``What is the BGM ?'' as TSC \textcircled{2} to ask the background music, which deviates the video content. So the TSC \textcircled{2} has the less semantic association with other TSCs, while TSC \textcircled{1} and TSC \textcircled{3} have a semantic edge.

\subsection{Topic Partitioning}
\label{3.2}
In this section, we will partition the topic of each TSC according to the semantic relationships in SAG. In our algorithm, the TSC that has the similar semantics and similar timestamps should belong to the same topic. However, the density of TSCs (number of TSCs per unit time) affects how users communicate. Therefore, we propose a dialogue-based cluster algorithm in Section \ref{dia} for the videos with sparse TSCs and a topic center-based cluster algorithm in Section \ref{tc} for the videos with dense TSCs.

\subsubsection{Dialogue-based Algorithm}

\label{dia}
First, we provide a dialogue-based algorithm. When the density of TSCs is low, the user can more clearly distinguish the content of each nearby TSC, and therefore is more likely for the user to reply to a specific TSC when posting the new one. Therefore, we cluster the TSCs according to the semantic relationship between each pair. The main idea is that the mean weight of edges in intra-topic is large while the mean weight of edges that link different topics is small, which satisfies community detection theory \cite{lancichinetti2008benchmark}.

Specifically, in the beginning, each TSC belongs to a unique topic. We use a unique set that only contains itself to achieve the objective. That is, for $v_{i}$, $v_{i}.S=\{i\}$. Then edges in set $E$ are sorted by descending order of weight. The new edge set $E^{'}=\{e_{1}^{'},e_{2}^{'},...,e_{k}^{'},...,e_{M}^{'}\}$ is obtained, where $e_{1}^{'}.w>e_{2}^{'}.w>...>e_{M}^{'}.w$. We process each edge from $e_{1}^{'}$ to $e_{M}^{'}$. For edge $e_{k}^{'}$, $S_{1}$ and $S_{2}$ represent the set $e_{k}^{'}.x.S$ and $e_{k}^{'}.y.S$. 
The set $S_1$ and $S_2$ should be merged if and only if TSCs in two sets discuss the similar topics. Therefore, we merge $S_{1}$ and $S_{2}$ if
\begin{equation}
S_{1} \neq S_{2}
\end{equation}
and
\begin{equation}
\frac{\sum\limits_{e_{p}.x, e_{p}.y \in {S_{1} \cup S_{2}}} e_{p}.w}{(|S_1 \cup S_2|)\cdot (|S_1 \cup S_2|-1)/2} > \rho_{d},
\end{equation}
where $\rho_{d}$ is the threshold of intra-cluster density. That is, we merge S1 and S2 only if the average edge weight of the their union is greater than the threshold. In this paper, disjoint-set (union-find set) algorithm \cite{tarjan1975efficiency} is used to merge the sets efficiently. When all the edges are solved, TSCs with high semantic similarity are merged into a topic, and the intra-cluster density of each subgraph is higher than the threshold. 

An example of dialogue-based topic partitioning is shown in Fig. \ref{fig:graph} (b1). The SAG constructed in Fig. \ref{fig:graph} (a) is finally partitioned into two topics marked as red and blue, and several noises marked as purple in Fig. \ref{fig:graph} (c) . The TSC ``Great player Messi!'' and ``Messi deserves MVP!" belong to the red topic, while the TSC ``What is the BGM ?'' is identified as a noise.

The full algorithm is shown in Algorithm \ref{dialougue}.

\begin{algorithm}[h]
 \caption{Dialogue-based algorithm}
 \label{dialougue}
 \renewcommand{\algorithmicrequire}{\textbf{Input}}
\renewcommand{\algorithmicensure}{\textbf{Output}}
 \begin{algorithmic}[1]
  \REQUIRE the edge set $E$
  \ENSURE the topic set of each time-sync comment
  \STATE sort $E$ by descending order of $e_{i}.w$, obtain $E^{'}$
  \FOR {$i$ = 1 to $M$}
    \STATE set $e_{i}^{'}.x.S$ as $S_{1}$, $e_{i}^{'}.y.S$ as $S_{2}$
    \IF {$\frac{\sum\limits_{e_{p}.x, e_{p}.y \in {S_{1} \cup S_{2}}} e_{p}.w}{(|S_{1}|+|S_{2}|)\cdot (|S_{1}|+|S_{2}|-1)/2} > \rho_{d}$}
        \STATE merge $S_{1}$ and $S_{2}$
    \ENDIF
  \ENDFOR
 \end{algorithmic}
\end{algorithm}

\subsubsection{Topic Center-based Algorithm}
\label{tc}
In the dialogue-based algorithm, we assume that TSCs are in the form of dialogues. However, when the density of TSCs is high, the user cannot clearly distinguish the content of each TSC, but only roughly distinguish the topic of these TSCs. Therefore, the user is more likely to reply to the entire topic instead of a specific TSC. The results of dialogue-based model will be disturbed by these situations. Therefore, we provide a Topic Center-based algorithm, which is inspired by Hierarchical Agglomerative Clustering \cite{pang2015unsupervised,murtagh2014ward,murtagh2008hierarchical,murtagh2012algorithms}.

Before proposing this algorithm,  the definition of topic center is given at first. As we defined in Section \ref{graph}, the set is used to describe the topic and each TSC can be express as an embedding vector by word2vec. The topic center is the average vectors of all TSCs within the topic. We use $S.center$ to express the topic center vector, and  $S.st$ and $S.ct$ to express the start time and center time of topic set $S$, respectively. Initially, each TSC belongs to a unique topic, so $v_{i}.S.center=vec_{i}$, $v_{i}.S.st=v_{i}.S.ct=t_{i}$, where $vec_{i}$ is the sentence embedding vector of TSC $i$. All these sets belong to $ST$, which is the universal set of topic sets.

Generally, this algorithm can be divided into two parts. (1) Find the nearest two topic centers. (2)  Merge the two topic centers. It is actually a Nearest Neighbor Search (NNS) problem \cite{bentley1975multidimensional,alstrup2000new}, where the k-d tree \cite{bentley1990k,friedman1977algorithm,bentley1975multidimensional} is one of the most effective methods. However, the analyses of binary search trees have found that the worst case time for range search in a k-dimensional k-d tree containing N nodes is given by the following equation \cite{lee1977worst}: $t_{worst}=O(k\cdot N^{1- \frac{1}{k}})$. Besides, the k-d tree has a larger constant. 

In this paper, we propose a greedy algorithm to solve this problem efficiently.  In the beginning, for each $S_{i} \in ST$, we find $S_{j}\in ST$ that
 \begin{equation}
 \label{match}
 S_{j}=\mathop{\argmax}_{j}{{\rm Affinity} (S_i,S_j)},
\end{equation}  where
 \begin{equation}
 \label{dist}
 {\rm Affinity} (S_i,S_j)=Sim(S_{i}.center,S_{j}.center) \cdot exp^{-\gamma_{t} \cdot (|S_{j}.ct-S_{i}.st|)}.
\end{equation}
The decay function is still added to avoid that the topics with large time interval are merged. 

We use $match_{i}$ to express the set that matches $S_{i}$ with maximum ${\rm Affinity} (S_i,match_{i})$ value $maxval_{i}$. And the pair $(S_{i},match_{i})$ is added to a queue ${\rm Aff}Q$, which is a priority queue where the pair $(S_{k},match_{k})$ with the maximum $maxval_{k}$ is the front. 

Each time, we take out the front pair $(S_{i},S_{j})$,  merging $S_{i}$ and $S_{j}$, and pop it, until ${\rm Affinity}(S_i,S_j)<\rho_{c}$. When merging sets, the following updates will be done: First, since $S_{i}$ and $S_{j}$ are merged, all pairs that contain $S_{i}$ or $S_{j}$, for instance $(S_{i},S_{u})$, should be deleted from ${\rm Aff}Q$. Then, these sets that matched $S_{i}$ or $S_{j}$ previously like $S_{u}$ are added into the update list $Ulist$. Next, the sets $S_{i}$ and $S_{j}$ are removed from $ST$, and a new set $S_{v}$ is added into $ST$ and $Ulist$,
where
 \begin{equation}
 \label{center}
 S_{v}.center=\frac{S_{i}.center\cdot |S_{i}|+S_{j}.center\cdot |S_{j}|}{|S_{i}|+|S_{j}|},
\end{equation}
\begin{equation}
 \label{st}
S_{v}.st=min(S_{i}.st,S_{j}.st),
\end{equation}
 and
 \begin{equation}
  \label{ct}
 S_{v}.ct=\frac{S_{i}.ct\cdot |S_{i}|+S_{j}.ct\cdot |S_{j}|}{|S_{i}|+|S_{j}|}.
 \end{equation}
 That is, the center time and the center vector of $S_v$ are the weighted average of $S_i$ and $S_j$, and the start time of $S_v$ is the minimum of $S_i$ and $S_j$. Finally, for each set $S_u\in Ulist$, we find a new $match_{u}$ according to Eq.(\ref{match}) in $ST$ to match it.

What is more, there exists a greedy optimization in the algorithm. Before giving the greedy optimization, we propose a lemma at first:\\
\begin{lemma}
For the set $S_{i}$, let $match_{i}=S_{j}$. Then the pair $(S_{i},S_{j})$ will never be solved in ${\rm Aff}Q$ if ${\rm Affinity}(S_i,S_j)<maxval_{j}$. 
\label{l1}
\end{lemma}

\begin{proof}
Since ${\rm Affinity}(S_i,S_j)<maxval_{j}$, we have $match_j=S_k \neq S_i$, and ${\rm Affinity}(S_i,S_j)<{\rm Affinity}(S_j,S_k)$. There exist two cases:

\textbf{Case $\romannumeral1$}: $match_{k}=S_{j}$
Then, in the priority queue ${\rm Aff}Q$, the pair $(S_{j},S_{k})$ will be solved before $(S_{i},S_{j})$ because ${\rm Affinity}(S_i,S_j)<{\rm Affinity}(S_j,S_k)$. Therefore, the pair $(S_{i},S_{j})$ will be removed from ${\rm Aff}Q$ when solving $(S_{j},S_{k})$.

\textbf{Case $\romannumeral2$}: $match_{k}=S_{p}\neq S_{j}$

Then we have ${\rm Affinity}(S_k,S_p)>{\rm Affinity}(S_j,S_k)$ (otherwise $match_{k}=S_{j}$). So the pair $(S_{k},S_{p})$ will be solved before $(S_{j},S_{k})$ in the priority queue ${\rm Aff}Q$. When solving $(S_{k},S_{p})$, $(S_{j},S_{k})$ will be removed, and set $S_{j}$ will find a new $match_{j^{'}}$ in $ST$. If $match_{j^{'}}=S_{i}$, then $(S_{i}, S_{j})$ is re-added into ${\rm Aff}Q$ (at that time, ${\rm Affinity}(S_i,S_j)=maxval_{j}$). Otherwise, $match_{j^{'}}=S_{q}\neq S_{i}$. In that case, the pair $(S_{j},S_{q})$ will be solved before $(S_{i},S_{j})$, and $(S_{i},S_{j})$ will be removed when solving $(S_{j},S_{q})$. Therefore, $(S_{i},S_{j})$ will always be removed and never be solved in any case.
\end{proof}

According to Lemma \ref{l1}, we propose the greedy optimization: for the set $S_{i}$, if $match_{i}=S_{j}$ and ${\rm Affinity}(S_i,S_j)<maxval_{j}$, then the pair $(S_{i},S_{j})$ is rejected and not added into ${\rm Aff}Q$.

The process of Topic Center-based Algorithm is described in Fig. \ref{fig:graph} (b2) and the clustering results are the same with the dialogue-based algorithm in this example showing in Fig. \ref{fig:graph} (c). The full algorithm is shown in Algorithm \ref{topiccenter}.
\begin{algorithm}[htbp]
\caption{Topic center-based algorithm}
\label{topiccenter}
 \renewcommand{\algorithmicrequire}{\textbf{Input}}
\renewcommand{\algorithmicensure}{\textbf{Output}}
\begin{algorithmic}[1]
  \REQUIRE the vectors and timestamp of time-sync comments
  \ENSURE the topic set of each time-sync comment
  \FOR {$i$ = 1 to $N$}
    \STATE $S_{i}.center=vec[i]$
    \STATE $S_{i}.st=t_{i}$
    \STATE $S_{i}.ct=t_{i}$
\ENDFOR
\FOR {$i$ = 1 to $N$}
  \STATE find $S_{j}=match_{i}$ using Eq.(\ref{match})
  \STATE calculate $maxval_{i}$ using Eq.(\ref{dist})
  \IF {($maxval_{j}\leq  maxval_{i}$) and ($maxval_{i}>\rho_{c}$)}
  	\STATE push the pair $(S_{i},match_{i})$ into ${\rm Aff}Q$
  \ENDIF
\ENDFOR
\WHILE{${\rm Aff}Q$ not empty}
	\STATE $(S_{x},S_{y})={\rm Aff}Q.front()$
	\STATE Remove all the pairs $(S_{x},S_{u})$ and $(S_{u},S_{x})$ in ${\rm Aff}Q$
	\IF {$S_{u}\in ST$ and $S_{u}\neq S_{y}$}
		\STATE add $S_{u}$ into $Ulist$
          \ENDIF
          	\STATE Remove all the pairs $(S_{v},S_{y})$ and $(S_{y},S_{v})$ in ${\rm Aff}Q$
	\IF {$S_{v}\in ST$ and $S_{v}\neq S_{x}$}
		\STATE add $S_{v}$ into $Ulist$
          \ENDIF
          \STATE calculate $S_{z}.center$,  $S_{z}.st$, and $S_{z}.ct$ and  using Eq.(\ref{center}), Eq.(\ref{st}), Eq.(\ref{ct})
          \STATE remove $S_{x}$ and $S_{y}$ from $ST$
          \STATE add $S_{z}$ into $ST$ and $Ulist$
          \WHILE {$Ulist$ not empty}
          		\STATE $S_{tmp}=Ulist.front()$
          	          \STATE find $S_{tj}=match_{tmp}$ using Eq.(\ref{match})
          	            \STATE calculate $maxval_{tmp}$ using Eq.(\ref{dist})
  	      	\IF {($maxval_{S_{tj}}\leq  maxvak_{tmp}$) and ($maxval_{tmp}>\rho_{c}$)}
  		 \STATE push the pair $(S_{tmp},match_{tmp})$ into ${\rm Aff}Q$
 		 \ENDIF
          \ENDWHILE
\ENDWHILE
\end{algorithmic}
\end{algorithm}

\subsection{Weight Distribution and Tag Extraction}
\label{3.4}
We partition the topic in Section \ref{3.2} and get the topic of each TSC. In this section, we will attribute weight to each TSC according to the influence of its topic and the relationship in the semantic graph.

The weight of a TSC is affected by its topic popularity, so we define the popularity of the TSC $i$ as:
\begin{equation}
P_{i}=\frac{|v_{i}.S|}{\sqrt[K]{|S_{1}|\cdot|S_{2}|...|S_{K}|}},
\label{eq:6}
\end{equation}
where $S_{j}$ $(j=1,2,...,K)$ is the $j-th$ topic in SAG, and $K$ is the total number of topics in SAG. Obviously those topics with fewer TSCs are more likely to be noises and have less weight. According to Eq.(\ref{eq:6}), noises will have small values of popularity.

Within the topic, a TSC which affects more TSCs and is affected by fewer TSCs should have a higher weight. In order to quantitatively measure the weight of the TSC in a topic, we design a graph iterative algorithm below.

An influence matrix $\mathbb{M}_{N \times N}$ is established at first to express semantic relations within each topic. For the elements in the matrix,
\begin{equation}
m_{i,j}=\begin{cases}
e_{i,j}.w &\text{if $v_{i}.S=v_{j}.S$}\\
0 &\text{if $v_{i}.S \neq v_{j}.S$}\end{cases}
\label{eq:2},
\end{equation}
we use $I_{i,k}$ to denote the influence value of $i-th$ TSC after $k$ iterations. For each TSC $i$, $I_{i,0}=1$ initially. Then in the $k-th$ turn of iteration, there are two steps as follows:
\begin{equation}
 I_{i.2k-1}=I_{i,2k-2}+\sum\limits_{j=i+1}^{n}m_{i,j}\cdot I_{j,2k-1},
\label{eq:3}
\end{equation}
and
\begin{equation}
I_{i,2k}=\frac{I_{i,2k-1}}{I_{i,2k-1}+\sum\limits_{j=1}^{i-1}m_{j,i}\cdot I_{j,2k}}.
\label{eq:4}
\end{equation}

In the $(2k-1)-th$ iteration, we increase the influence value of TSC $i$ based on the values of TSCs that affected by TSC $i$. We know that a TSC only affects the TSCs lagging behind it, so the TSCs are processed from $v_{N}$ down to $v_{1}$. That is, before we process TSC $i$, all the TSCs $j$ that $t_{j}>t_{i}$ have been processed.
In the $(2k)-th$ iteration, we reduce the influence value of TSC $i$ based on the values of the TSCs that affect TSC $i$. Contrary to the $(2k-1)-th$ iteration, we process the TSCs from $v_{1}$ to $v_{N}$ in the $(2k)-th$ iteration. 

The iteration process of SAG in Fig. \ref{fig:graph} (c) is shown in Fig. \ref{fig:iter2}. Fig. \ref{fig:iter2} (a) shows the calculation of the last two nodes (marked as red) that need to be processed in the $(2k-1)-th$ iteration (ignore the noise node $I_2$), where the orange edges express their out-degree edges. While Fig. \ref{fig:iter2} (b) shows the calculation of the last two nodes (marked as red) that need to be processed in the $(2k)-th$ iteration (ignore the noise node $I_{10}$), where the green edges express their in-degree edges.

\begin{figure}[htbp]
\centering
\includegraphics[width=1\linewidth]{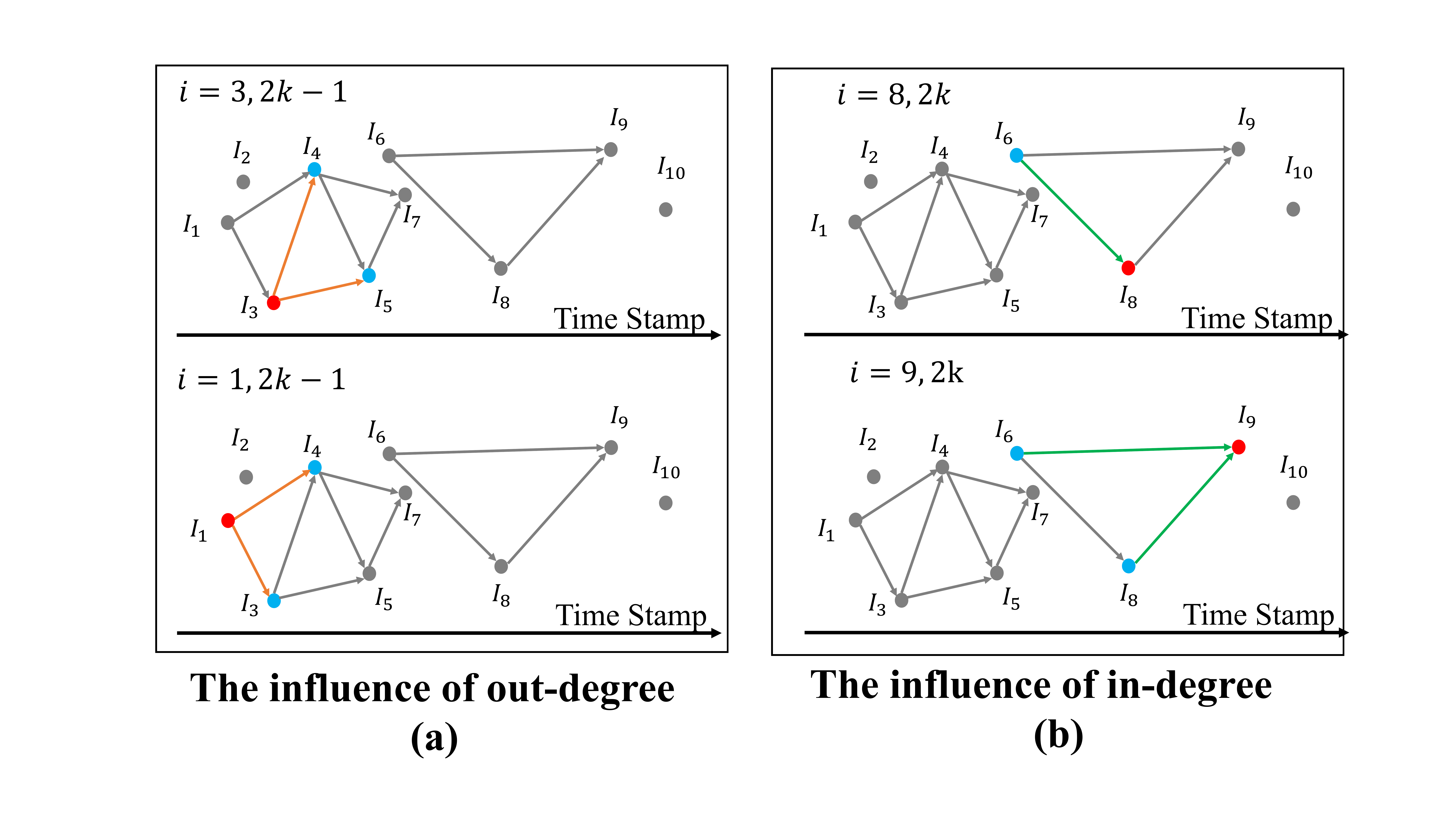}
\caption{The Iteration process of SAG in Fig. \ref{fig:graph} (c).}\label{fig:iter2}
\end{figure}

The converged influence values of the 10 TSCs in Fig. \ref{fig:graph} (c) is shown in Fig. \ref{fig:iter}. After 20 iterations, all TSCs converge to the interval $[0,1]$.

\begin{figure}[htbp]
\centering
\includegraphics[width=0.8\linewidth]{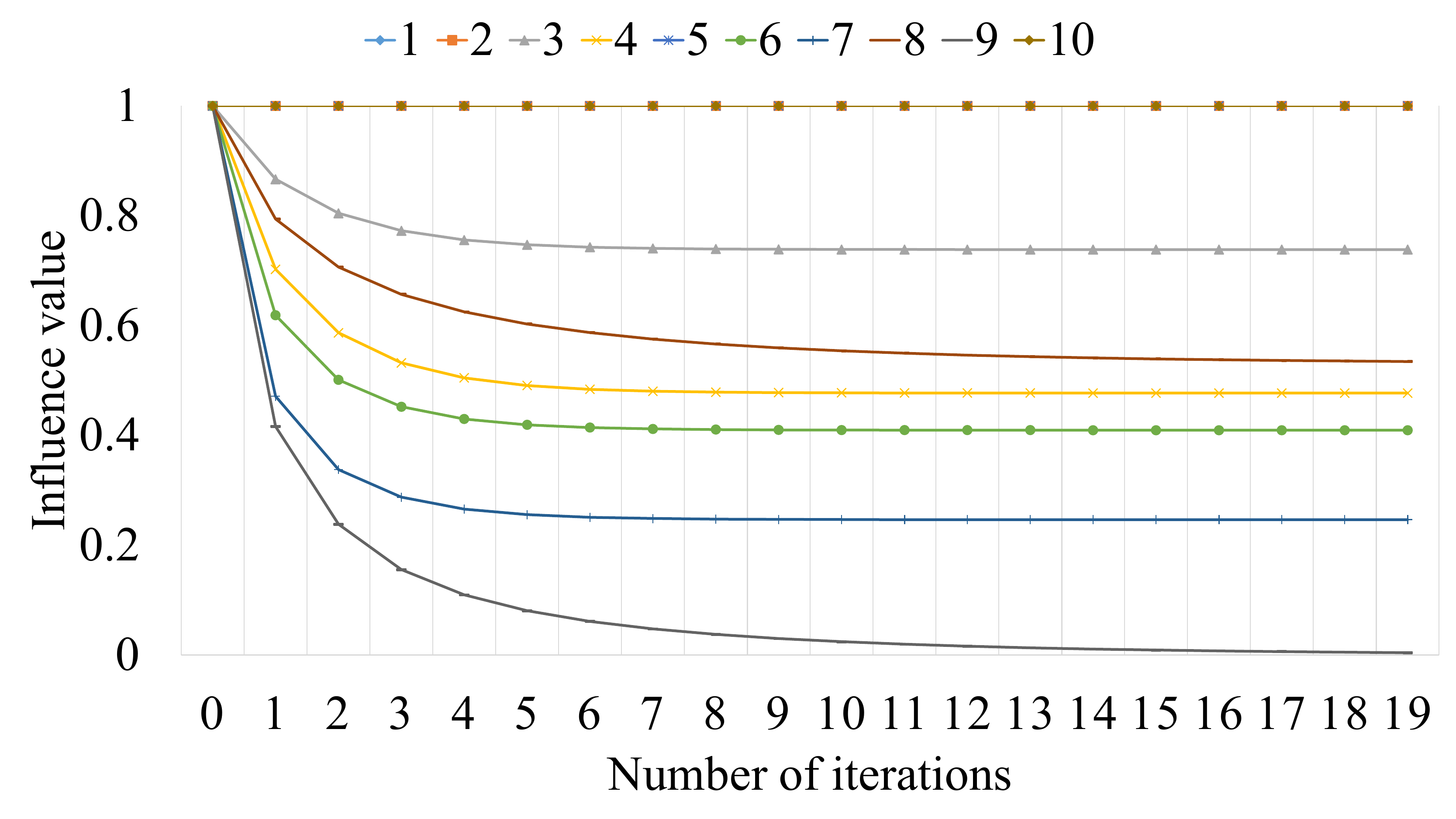}
\caption{The influence value of TSCs in Fig. \ref{fig:graph}}
\label{fig:iter}
\end{figure}

To combine the popularity and the influence value, the weight of TSC $i$ is obtained by
\begin{equation}
W_{i}=P_{i}\cdot I_{i}^{T},
\end{equation}
where $T$ is the number of turns of iterations and depends on the number of nonzero elements in the matrix $\mathbb{M}_{N \times N}$.
Therefore, the weight of each word is formulated as below:
\begin{equation}
SW-IDF_{i}=\sum\limits_{j}W_{j}\cdot IDF_{i},
\label{eq:5}
\end{equation}
where $j$ denotes the TSC that word $i$ appears and $IDF_{i}$ is the inverse document frequency as defined in TF-IDF method. We extract words with the highest SW-IDF value as video tags. After the above steps, those words which appear in the TSCs that are popular and have high impact will be extracted as tags. The complete algorithm is shown in Algorithm \ref{al:2}.
\begin{algorithm}[h]
 \caption{EXTRACTING TAGS BY SW-IDF}
 \label{al:2}
 \renewcommand{\algorithmicrequire}{\textbf{Input}}
\renewcommand{\algorithmicensure}{\textbf{Output}}
 \begin{algorithmic}[1]
  \REQUIRE Semantic Association Graph
  \ENSURE Tags of video
 \STATE Assign time-sync comments to a set by Algorithm \ref{dialougue} or Algorithm \ref{topiccenter}
  \STATE Calculate the influence matrix $\mathbb{M}_{N \times N}$ using Eq.(\ref{eq:2})
  \FOR {$i$ = 1 to $N$}
    \STATE $I_{i}^{0}=1$
    \STATE Calculate the popularity of TSC $i$ using Eq.(\ref{eq:6})
  \ENDFOR
  \FOR {$k$ = 1 to $T$}
    \FOR {$i$ = N downto $1$}
    \STATE Calculate $I_{i}^{2k-1}$ using Eq.(\ref{eq:3})
    \ENDFOR
    \FOR {$i$ = 1 to $N$}
    \STATE Calculate $I_{i}^{2k}$ using Eq.(\ref{eq:4})
    \ENDFOR
  \ENDFOR
  \STATE Calculate the SW-IDF of each word using Eq.(\ref{eq:5})
  \STATE Select words with max SW-IDF as video tags
 \end{algorithmic}
\end{algorithm}

\subsection{Complexity Analysis}
\label{ca}
In this section, we analyze the time complexity and the space complexity of each algorithm.

In Algorithm \ref{dialougue}, the time complexity of the edge sorting algorithm in line 1 is $O (MlogM)$ by using quicksort, and the space complexity is $O(M)$. The amortized time complexity of merging sets by disjoint-set is $O(\alpha(n))$ \cite{tarjan1979class} and the space complexity is $O(N)$, where $\alpha(n)$ is the inverse Ackermann function that $\alpha(n)<5$. So the total time complexity of Algorithm \ref{dialougue} is $O(MlogM+M\alpha(N))$, and the space complexity is $O(M + N)$.

In Algorithm \ref{topiccenter}, the time complexity of initialization from line 1 to line 12 is $O(N^{2})$, and the space complexity is $O(N)$. In ${\rm Aff}Q$, the number of times of merge-operation is limited to $N-1$ (because there are at most $N $sets), and the amortized removal operation is limited to 1 each merge-operation. For each merge operation, the lookup operation and remove operation can be dealt in $O(N)$ by naive algorithm, or $O(logN)$ by binary balance tree \cite{bentley1975multidimensional}. The worst complexity of total Algorithm \ref{topiccenter} is $O(N^{2})$. The total  space complexity is just $O(N)$.

In Algorithm \ref{al:2}, the time complexity is $O(T\cdot N^{2})$ and the space complexity is $O(M+N^2)$ apparently. In our SAG, $M<N^{2}$ because two TSCs with a negative semantic similarity do not have an edge. Therefore, $O(MlogM+M\alpha(N)) < N^{2})$ in the true TSC data, and the dialogue-base algorithm has a more efficient time complexity than the topic center-based algorithm.

\section{Experimental Study}
\label{sec:4}
In this section, we verify the effectiveness of our proposed method by comparing with four unsupervised methods of keyword extraction. The datasets are crawled from AcFun (www.acfun.cn) and Bilibili. We provide the necessary parameters in our algorithms in Section \ref{exset} and then analyze the performance of our algorithms on video tag extraction in Section \ref{results}.
\subsection{Experimental Setup and Datasets}
\label{exset}
We crawl TSCs from two famous Chinese time-sync comments video websites AcFun and Bilibili. The raw TSC texts are full of noises, so we manually remove non-textual TSCs (such as emojis) and establish a set of mapping rules for network slang, which will be substituted by their real meaning in the text. For instance, 233... (2 followed by several 3) means laughter, 666... (several 6) means playing games very well. After that, we segment the words and remove the anomaly symbol (the symbolic expression, such as a smiley face (\^{}\_\^{}) ) in TSCs by an open-source Chinese-language processing toolbox Jieba \footnote{https://github.com/fxsjy/jieba}. To analyze the algorithms from different aspects, we collected two datasets. To be specific, in the first dataset (called it D1), totally 287 videos with 227,780 comments are collected randomly from music, sports, and movie. To set the hyper-parameters in this paper, we select 167 videos with 126,146 TSCs for the validation set and 120 videos with 101,634 comments for the test set. In the second dataset (called it D2), totally 180 videos with 569,996 comments are collected from Japanese anime. We use D1 to compare our algorithms with baselines, and use D2 to accurately analyze the effects of the two algorithms we proposed at different densities. 

We define the density of TSCs as the average number of TSCs per minute. In D2, we divide the density into 5 levels: 0-30, 30-60, 60-90, 90-120 and more than 120 (the intervals are left-closed and right-open). More details include the length of the video, total number of TSCs, density and the number of videos about test set are shown in Table \ref{table1} for D1 and Table \ref{density} for D2.

\begin{table}[htbp]
\caption{Data Description Table for D1}
\centering
\begin{tabular}{|c|c|c|}
\hline
&Validation set& Test set\\
\hline
Total length (minute) & 1,573.29 & 1,441.38\\
\hline
Total TSCs number& 126,146 & 101,634\\
\hline
Density & 80.18 & 70.51\\
\hline
Total video number & 167 & 120\\
\hline
\end{tabular}
\label{table1}
\end{table}

\begin{table}[htbp]
\caption{Data Description Table for D2}
\centering
\begin{tabular}{|c|c|c|c|c|c|}
\hline
&0-30 & 30-60& 60-90 & 90-120 & >120\\
\hline
Total length (minute) & 644.37 & 433.01 &855.40 & 883.61 & 1,221.55\\
\hline
Total TSCs number& 11,489 & 19,368 & 60,152 & 99,671 & 379,316\\
\hline
Density & 17.83 & 43.72 & 70.32 & 112.80 & 310.52\\
\hline
Total video number & 29 & 21 & 37 & 42 & 51\\
\hline
\end{tabular}
\label{density}
\end{table}

We select two undergraduate students and one Ph.D. student as volunteers. For each video, each volunteer chooses 15 words from TSCs and votes them as video tags. The words with two or more votes are selected as the standard tags.  Therefore, the number of standard tags per video is different. Moreover, the order of these tags is determined by the number of votes at first. TSCs with more votes rank in front. When the number of votes is the same, the order is determined by the Ph.D. student. \footnote{The code of our algorithm is uploaded  to https://github.com/sdq11111/SAG.}

In Section \ref{graph}, we use the word2vec method get the embedding vectors of TSCs. In this paper, we choose the skip-gram model of word2vec to pre-train the word embedding vectors and the training algorithm is hierarchical softmax, because both skip-gram model and hierarchical softmax algorithm are better for infrequent words \cite{mikolov2013distributed}, which is more relevant to the features of the TSCs. We use gensim \footnote{https://radimrehurek.com/gensim/models/word2vec.html} to train the model, and the training data is crawled from Bilibili with the TSCs of 6,743,912 words. Since we have sufficient training corpus, the dimension $d$ of word2vec is set to 300 as \cite{li2018analogical}.

To further prove the rationality of using the word2vec to calculate the similarities of the TSCs, we use several traditional unsupervised learning and other word embedding methods to calculate the semantic similarities, \emph{i.e.}
\begin{enumerate}[(1)]
\item LDA, a famous topic model based method, Latent Dirichlet Allocation \cite{blei2003latent}.
\item PPMI, a co-occurrence probability based distributional model, Positive Pointwise Mutual Information \cite{levy2014linguistic}
\item HowNet, a HowNet hierarchical sememe tree based approach \cite{wu2012chinese},  where HowNet \cite{dong2003hownet} is a common-sense knowledge base unveiling inter-conceptual relations and inter-attribute relations of concepts. 
\item GLoVe, a famous word embedding method, Global Vectors for Word Representation \cite{pennington2014glove}.
\end{enumerate}

We test the top 10 tag extraction results using the above methods to calculate the similarity and build the graph on the verification set (the hyper-parameters used in the experiment are discussed later). In this paper, we use F1-score and MAP (Mean Average Precision, which is the mean of the average precision scores for each query \cite{zhu2004recall}) to measure the performance of tag extraction. The results are shown in Table \ref{embedding}.

\begin{table}[htbp]
\centering
\caption{The effect of semantic similarity calculation method on the results}
\begin{tabular}{|c|c|c|c|c|}
\hline
Method & F1 (dialogue) & MAP (dialogue) & F1 (topic center) & MAP (topic center)\\
\hline
LDA & 0.3625 & 0.3372 & 0.3641 & 0.3224\\
\hline
PPMI & 0.3919 & 0.3705 & 0.4101 & 0.3806\\
\hline
HowNet& 0.3537 & 0.3423 & 0.3468 & 0.3194\\
\hline
GLoVe & 0.4045 & 0.4012 & 0.4202 & 0.4079\\
\hline
Word2Vec & 0.4183 & 0.4041 & 0.4342 & 0.4160\\
\hline
\end{tabular}
\label{embedding}
\end{table}

The experimental results show that, in the verification set, Hownet performs the worst among the baselines because of the limited number of word lists. LDA also performs poorly because it is not good at handling short texts. Among the word embedding based methods PPMI, GLoVe, and word2vec, word2vec performs best, which indicates that the fully trained word2vec method has better robustness and is more suitable for calculating the similarity of the TSCs.

What is more, in our algorithm, three parameters need to be determined, \emph{i.e.}, the threshold of intra-cluster density $\rho_{d}$ and $\rho_{c}$, and the attenuation coefficient $\gamma_{t}$. The $\rho_{d}$ and $\rho_{c}$ control the accuracy of topic clustering. The $\gamma_{t}$ is the attenuation coefficient of the interval between time-sync comments, which controls the value of the edge weights in the graph. 

We first fix $\gamma_{t}$ and adjust the values of $\rho_{d}$ and $\rho_{c}$ so that the F1-score and MAP in the verification set are optimal. Then, we select the optimal $\rho_{d}$ and $\rho_{c}$ and re-adjust $\gamma_{t}$ so that the F1-score and MAP in the verification set is optimal. In Bilibili video site, the default time for each TSC to appear on the screen is 10 seconds. Therefore, we assume that the semantic half-life of each TSC is 5 seconds, and calculate the initial $\gamma_{t}=-ln0.5/5\approx0.14$ according to Eq. (\ref{decay}).

\begin{figure}[htbp]
\centering
\begin{minipage}[t]{0.49\linewidth}
\centering
\includegraphics[width=1\linewidth]{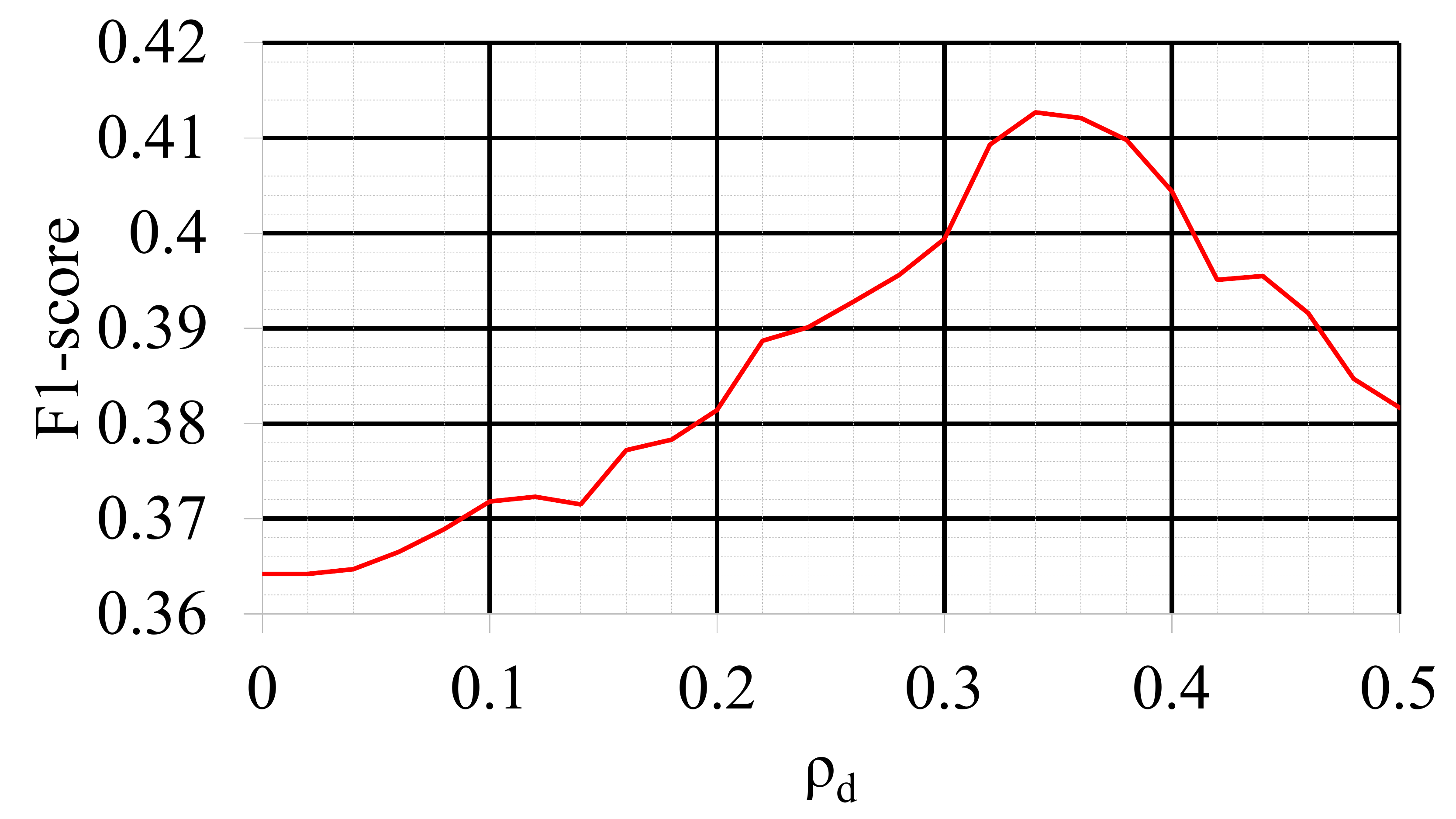}
\caption{The effect of threshold $\rho_{d}$ in F1-score}
\label{fig:diaf1}
\end{minipage}
\begin{minipage}[t]{0.49\linewidth}
\centering
\includegraphics[width=1\linewidth]{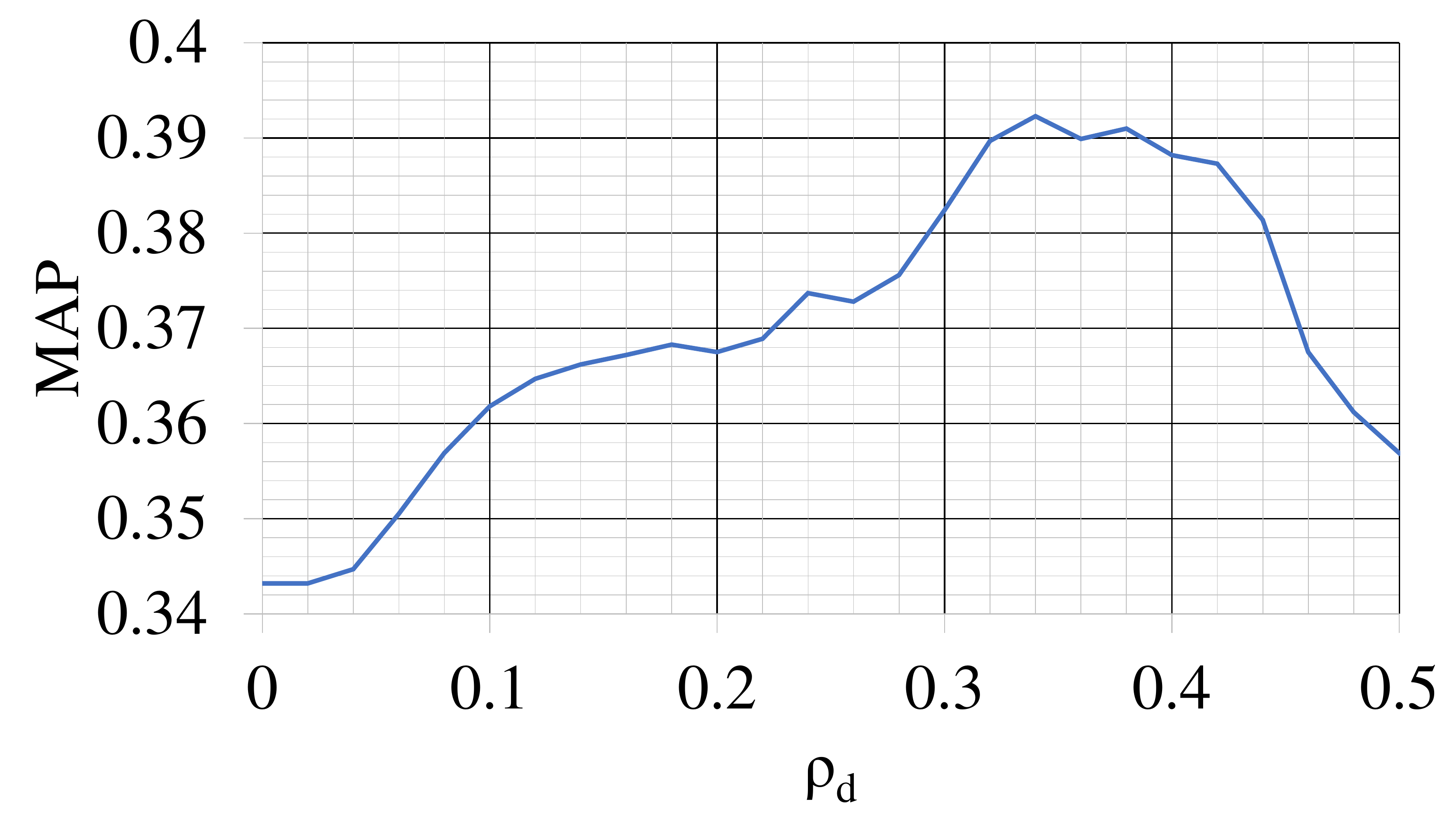}
\caption{The effect of threshold $\rho_{d}$ in MAP}
\label{fig:diamap}
\end{minipage}
\end{figure}

\begin{figure}[htbp]
\centering
\begin{minipage}[t]{0.49\linewidth}
\centering
\includegraphics[width=1\linewidth]{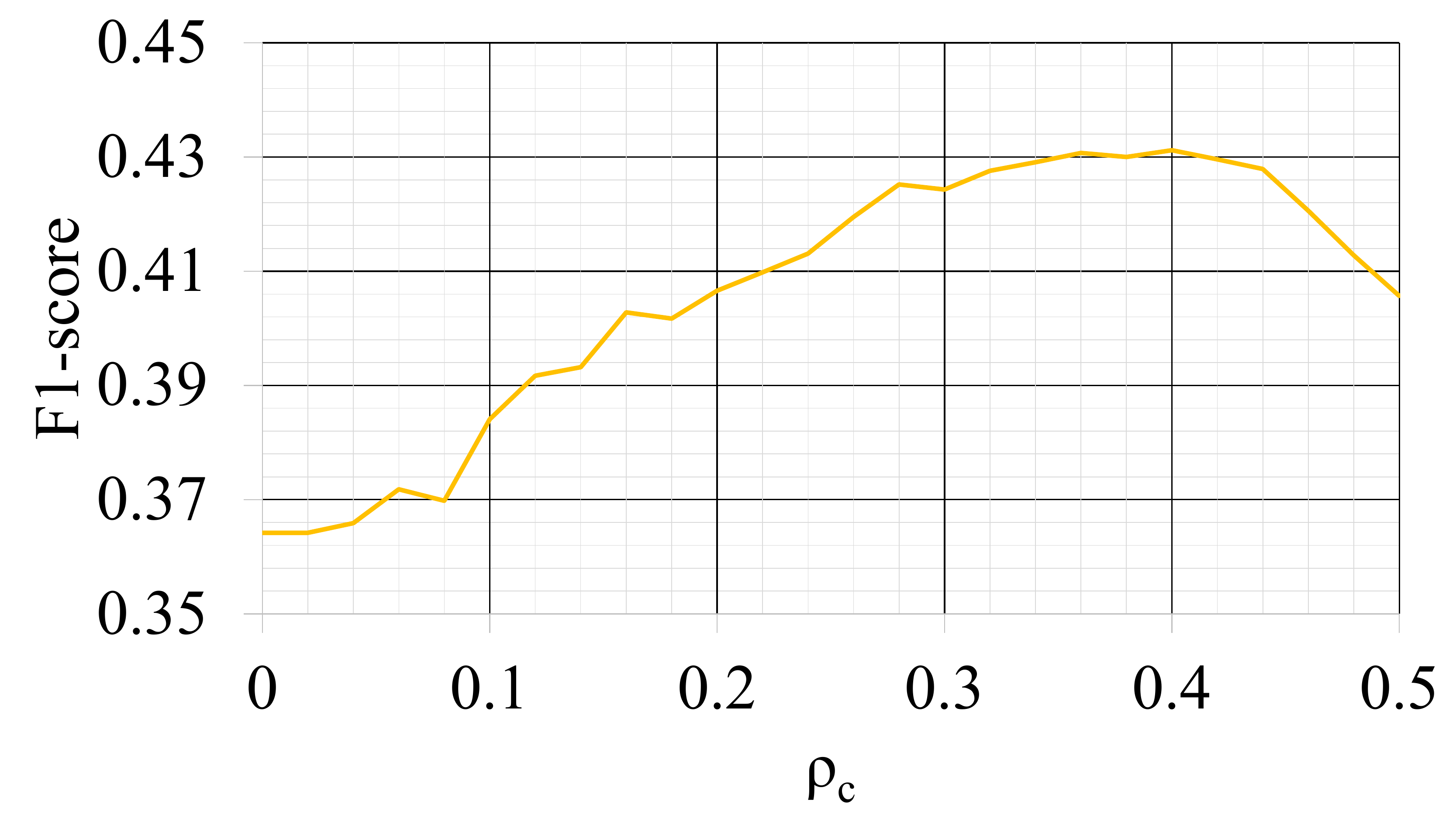}
\caption{The effect of threshold $\rho_{c}$ in F1-score}
\label{fig:topf1}
\end{minipage}
\begin{minipage}[t]{0.49\linewidth}
\centering
\includegraphics[width=1\linewidth]{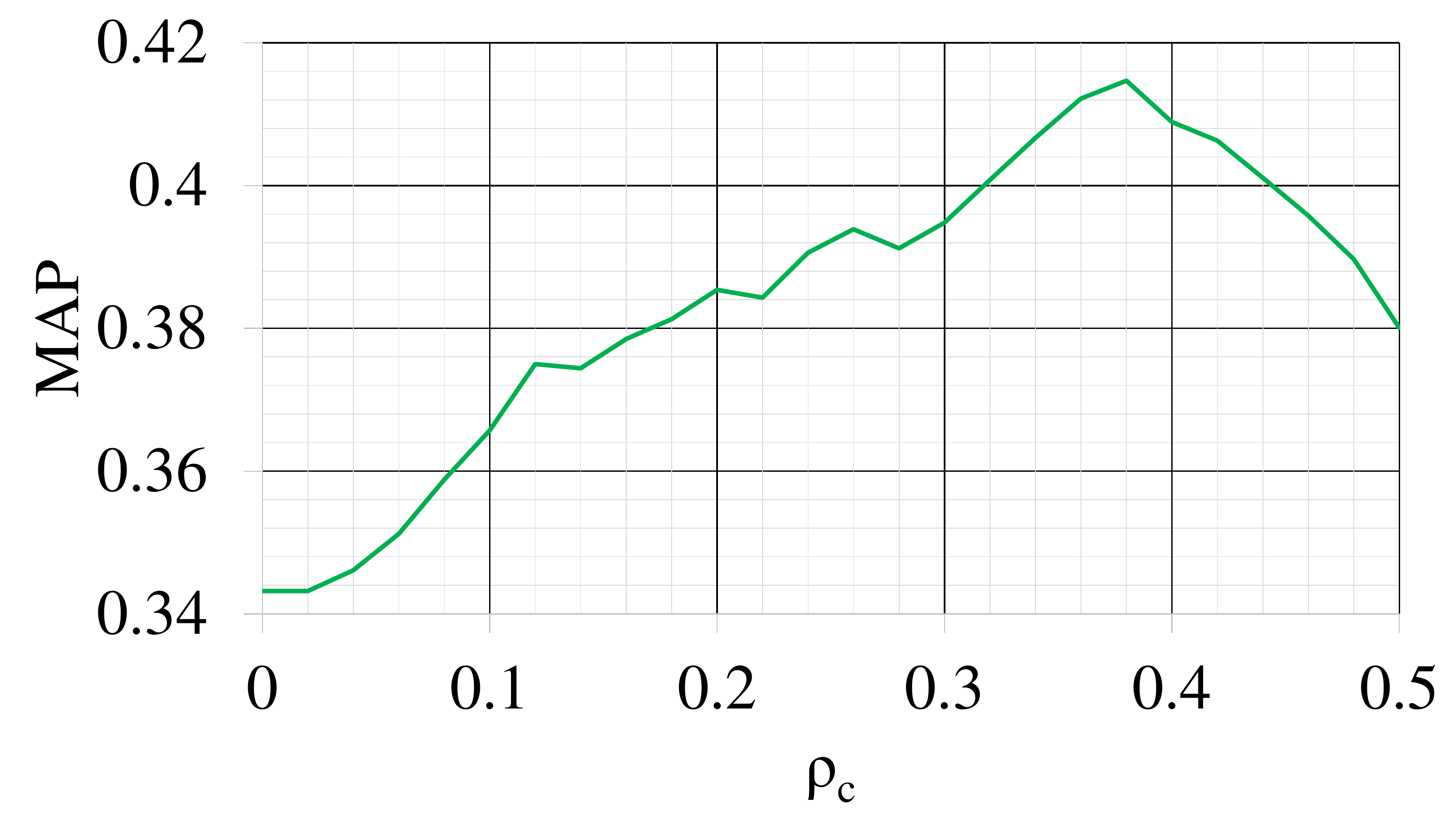}
\caption{The effect of threshold $\rho_{c}$ in MAP}
\label{fig:topmap}
\end{minipage}
\end{figure}

To determine $\rho_{d}$, we fix $\gamma_{t}=0.14$, adjusting $\rho_{d}$ from 0 to 0.5 in 0.02 steps and observe the F-score and MAP of Top 10 tagging results generated by the dialogue-based algorithm. The results of F1-score and MAP in the validation set are shown in Fig.\ref{fig:diaf1} and Fig.\ref{fig:diamap}, respectively. Both in F1-score and MAP, $\rho_{d}$ gains better results in the range of 0.32 to 0.38 and get optimal performance at 0.34. Therefore, we choose $\rho_d = 0.34$ for the following experiments.

To determine $\rho_{c}$, we also fix $\gamma_{t}=0.14$, adjusting $\rho_{c}$ from 0 to 0.5 in 0.02 steps and observe the F-score and MAP of Top 10 tagging results generated by the topic center-based algorithm. The results of F1-score and MAP in the validation set are shown in Fig.\ref{fig:topf1} and Fig.\ref{fig:topmap}, respectively. For F1-score, $\rho_c$ gains better results in the range of 0.34 to 0.42 and get optimal performance at 0.40. For MAP, $\rho_c$ gains better results in the range of 0.34 to 0.40 and get the optimal performance at 0.38. Considering both F1-score and MAP, we choose $\rho_c = 0.38$ for the following experiments.

\begin{figure}[htbp]
\centering
\begin{minipage}[t]{0.49\linewidth}
\centering
\includegraphics[width=1\linewidth]{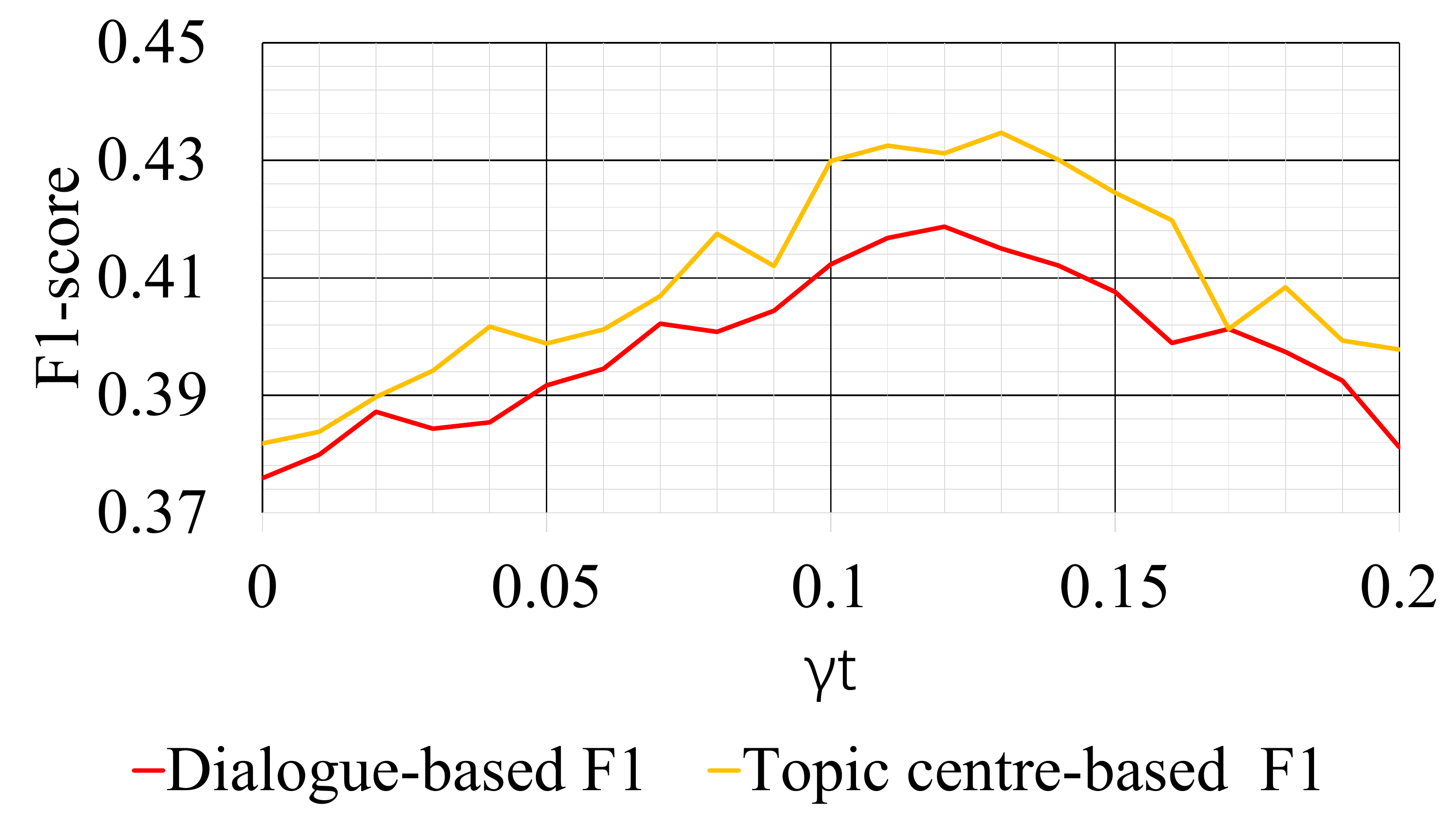}
\caption{The effect of attenuation coefficient $\gamma_{t}$ on F1-score}
\label{fig:timef}
\end{minipage}
\begin{minipage}[t]{0.49\linewidth}
\centering
\includegraphics[width=1\linewidth]{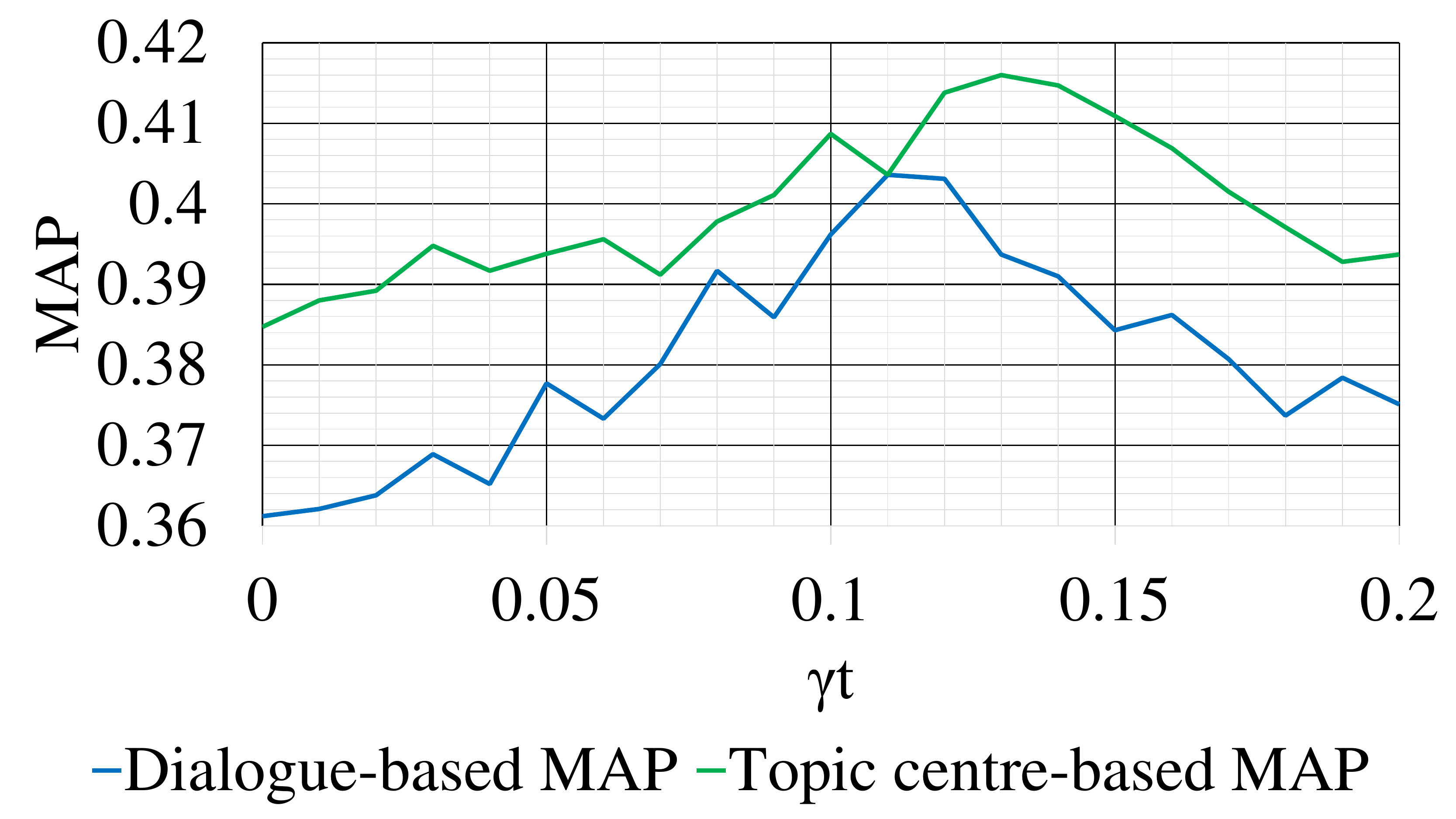}
\caption{The effect of attenuation coefficient $\gamma_{t}$ on MAP}
\label{fig:timem}
\end{minipage}
\end{figure}

With the optimal $\rho_{d}$ and $\rho_{c}$ obtained before, we re-adjust $\gamma_{t}$ from 0 to 0.2 in steps 0.01, and observe the F-score and MAP of video tags generated by our algorithms. The results of F1-score and MAP in the validation set are shown in the Fig. \ref{fig:timef} and Fig. \ref{fig:timem}. For the dialogue-based algorithm, $\gamma_{t}$ gains better performance in the range of 0.10 to 0.13 and gets optimal performance at 0.12 for F1-score and 0.11 for MAP. For topic the center-based algorithm, $\gamma_{t}$ gains better performance in the range of 0.10 tp 0.14 and gets optimal performance at  0.13 for both F1-score and MAP. To take comprehensive consideration of both F1-score and MAP, we choose $\gamma_{t}=0.12$ for the dialogue-based algorithm, and $\gamma_{t}=0.13$ for the topic center-based algorithm in the following experiments. In fact, when $\gamma_{t}=0$, the semantic association graph is independent of time; when $\gamma_{t}=+ \infty$, all weights of edge equal to 0, and our model is equivalent to TF-IDF.

Besides, the number of iterations $T$ also needs to be determined. We count the number of iterations when algorithms converge at different densities (we consider the algorithm converges when the average of $\frac{|I_{i,k}-I_{i-1,k}|}{I_{i-1,k}}<5\%$), the results are shown in Table \ref{T}.

\begin{table}[htbp]
\centering
\caption{The number of iterations when algorithms converged at different densities}
\begin{tabular}{|c|c|c|c|c|c|}
\hline
 &0-30 &31-60 & 60-90 & 90-120 & $>$120\\
\hline
Dialogue &7.32 & 13.59 & 27.59 & 35.15 & 43.82\\
\hline
Topic center & 6.89& 14.92 & 23.15 & 31.42 & 45.62\\
\hline
\end{tabular}
\label{T}
\end{table}

As shown in Table \ref{T}, when the density of TSCs is low, the SAG generated by two algorithms is sparse, and therefore the number of iterations is few. As the density increases, the SAG becomes dense and the number of iterations increases. To simplify, we choose $T=50$ in the experiment.

\subsection{Results}
\label{results}
In this section, we first use D2 to analyze the clustering effect of the two algorithms we proposed at different densities. Then, we use the test set of D1 to verify the effectiveness of the greedy optimization we proposed, and compare our algorithms with the existing methods TF-IDF, TextRank \cite{mihalcea2004textrank}, BTM \cite{yan2013biterm} GSDPMM \cite{yin2014dirichlet,yin2016model}, and TPTM \cite{wu2014crowdsourced}.

In the beginning, an experiment was designed to compare the clustering effect of the two algorithms. Given a set of topics $ST=\{S_{1},S_{2},...,S_{K} \}$, two distance scores are introduced \cite{yan2013biterm}.

\textbf{Average Intra-Cluster Distance:}
\begin{equation}
IntraDis(S)=\frac{1}{K}\sum_{k=1}^{K}\left [ \sum_{ \substack{v_{i},v_{j}\in S_{k} \\i \neq j}}
\frac{2\cdot {\rm Affinity}(v_{i},v_{j})}{|S_{k}| |S_{k}-1|}\right ]
\end{equation}

\textbf{Average Inter-Cluster Distance:}
\begin{equation}
InterDis(S)=\frac{1}{K(K-1)} \sum_{ \substack{S_{k},S_{k'}\in ST \\ k\neq k'}}\left[\sum_{v_{i}\in S_{k}}\sum_{v_{j}\in S_{k'}}\frac{{\rm Affinity}(v_{i},v_{j})}{|S_{k}||S_{k'}|}  \right]
\end{equation}

Since we use ${\rm Affinity}$ function to calculate the semantic similarity between two topics, where the higher the similarity is, the greater the function value is. Intuitively, if the Average Intra-Cluster Distance is high and the Average Inter-Cluster Distance is low, then the algorithm has a great clustering effect. Therefore, we calculate  
\begin{equation}
H=\frac{IntraDis(ST)}{InterDis(ST)}
\end{equation}
to evaluate the quality of clustering algorithms as \cite{guo2011intent,bordino2010query}. 

Due to the time decay function in the semantic association graph, the $H$ value, the IntraDis and the topic number (cluster number) of the videos vary greatly with the video duration. Therefore, we do not calculate the average value of all the videos directly but define an $H-hit$ score instead. That is, for each video, we compare the $H$ score obtained by the two cluster algorithms, and the algorithm with the larger $H$ score obtains a hit. The H-hit that the dialogue-based algorithm gets is called D-Hit, and the H-hit that the topic center-based algorithm obtains is called T-Hit. 
\begin{figure}[htbp]
\centering
\includegraphics[width=0.7\linewidth]{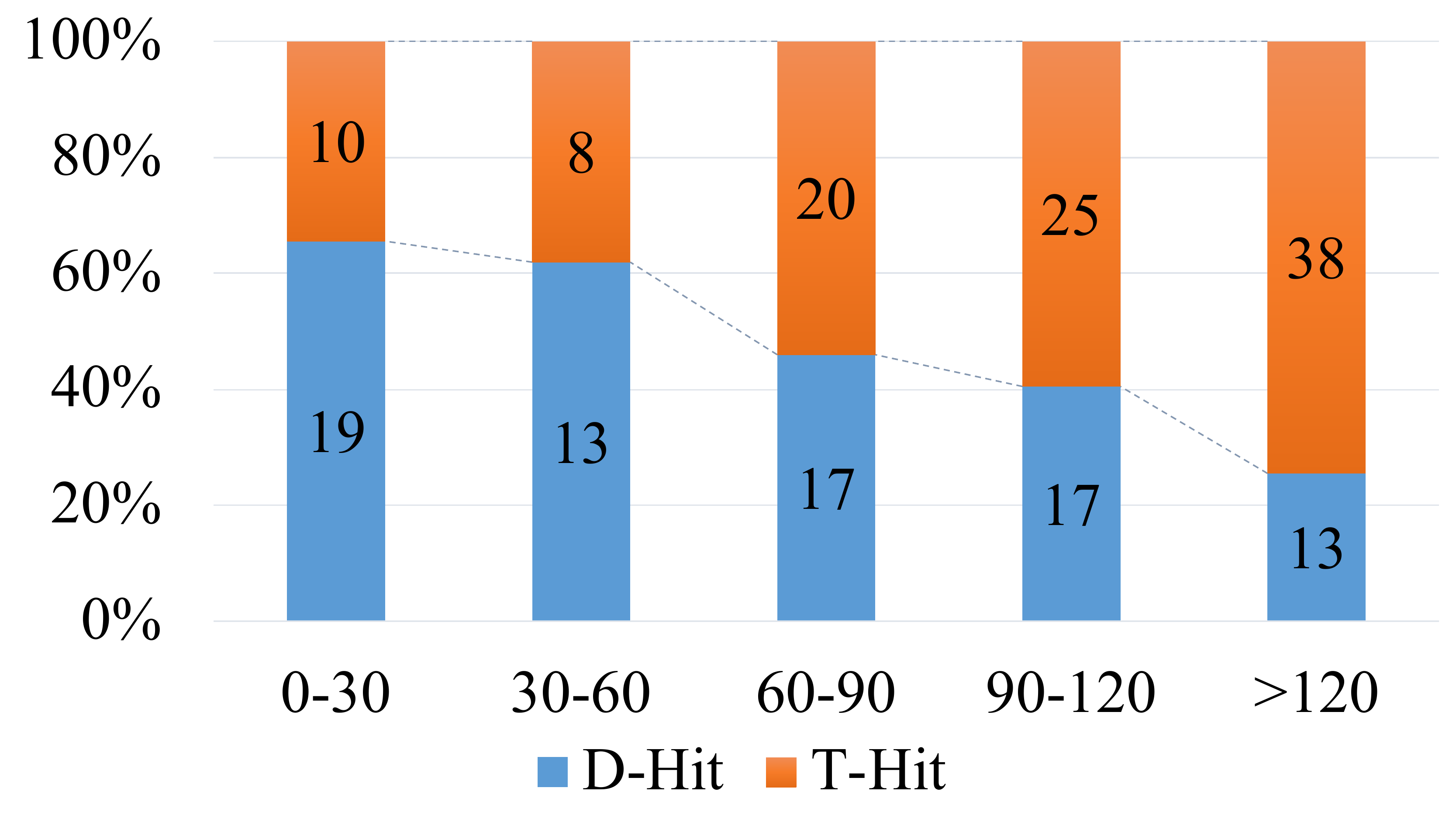}
\caption{The comparison of two clustering algorithms}
\label{fig:hit}
\end{figure}

The results are shown in Fig. \ref{fig:hit}. The dialogue-based algorithm performs better when the density is lower than 60. As the density increases and exceeds 60, the topic center-based algorithm performs better than the dialogue-based model. Moreover, we directly compare the top 10 tag extraction results of two clustering algorithms at different densities. The results are shown in Table \ref{self-density}.

\begin{table}[htbp]
\centering
\caption{The tag extraction results at different densities.}
\begin{tabular}{|c|c|c|c|c|c|}
\hline
 &0-30 &31-60 & 60-90 & 90-120 & $>$120\\
\hline
Dialogue F1-score&0.4357 & 0.4412 & 0.4219 & 0.4108 & 0.4383\\
\hline
Dialogue MAP& 0.3742 & 0.4027 &0.4615 & 0.4013 & 0.4872\\
\hline
Topic center F1-score& 0.4139& 0.4276 & 0.4275 & 0.4216 & 0.4433\\
\hline
Topic center MAP& 0.3615 & 0.3988& 0.4747 & 0.4077 & 0.5093\\
\hline
\end{tabular}
\label{self-density}
\end{table}

The tag extraction results are similar to Fig. \ref{fig:hit}. From Fig. \ref{fig:hit} and Table \ref{self-density}, we can conclude that the dialogue-based algorithm is better for videos with a density lower than 60, while topic center-based algorithm has significant advantages for videos with the density higher than 60, which fits our assumptions in Section \ref{3.2}. Based on the conclusions above, in the test set of D1, we consider the videos with the density of TSCs greater than 60 as high-density videos, and others are low-density videos. Then, the test set in D1 is divided into two parts: videos with high-density TSCs and with low-density TSCs. The details are shown in Table \ref{d1}.

\begin{table}[htbp]
\caption{Data Description Table for the test set of D1}
\centering
\begin{tabular}{|c|c|c|}
\hline
&High-density & Low-density\\
\hline
Total length (minute) & 124.58 & 1316.80\\
\hline
Total TSCs number& 41,556 & 60,078\\
\hline
Density & 333.56 & 45.62\\
\hline
Total video number & 89 & 31\\
\hline
\end{tabular}
\label{d1}
\end{table}

We use the data in Table \ref{d1} to verify the effectiveness of greedy optimization we proposed in Section \ref{tc}. Specifically, we run the code of Algorithm \ref{topiccenter} for 10 times, counting the running time from line 6 to line 34, with and without the greedy optimization (in line 9), respectively. The experiment platform we used is one MacBook Pro 13-inch, 2.9 Ghz Inter Core i5, 8GB 2133MHz LPDDR3 with single thread. We add up the total time of all the samples (since the single sample only runs for a short time). The average time of 10 runs is shown in Table \ref{greedy}.

\begin{table}[htbp]
\caption{Validation of greedy optimization}
\centering
\begin{tabular}{|c|c|c|}
\hline
&High-density & Low-density\\
\hline
Topic center only& 7.671 & 10.725\\
\hline
Topic center with greed& 6.905 & 10.060\\
\hline
\end{tabular}
\label{greedy}
\end{table}

The results show that the greedy optimization reduces 9.99\% running time of high-density data and 6.20\% of low-density data, respectively, which verifies the effectiveness of our greedy algorithm.

Then, we compare our algorithm with different existing methods using the test set of D1. To evaluate the performance of the proposed video tag extraction algorithm, we compare our method with 5 unsupervised keyword extraction methods, \emph{i.e.},
\begin{enumerate}[(1)]
\item TF-IDF, a classical keyword extraction algorithm.
\item TX, a graph-based text ranking model, textrank \cite{mihalcea2004textrank}, which is inspired by PageRank.
\item BTM, a topic model based algorithm, Biterm Topic Model \cite{yan2013biterm}, which is the improvement of LDA \cite{blei2003latent} for short texts. The number of topics is 20 in this experiment.
\item GSDPMM, a collapsed Gibbs Sampling algorithm for the Dirichlet Process Multinomial Mixture model \cite{yin2014dirichlet,yin2016model}, which has good performance when dealing with short texts. We set $\alpha = 0.1 * D$ ($D$ is the number of documents in the dataset), $K = 1$, and $\beta = 0.02$ in this experiment.
\item TPTM, a Temporal and Personalized Topic Model \cite{wu2014crowdsourced}, which is the first work on automatic TSC tagging. All parameters are set in accordance with \cite{wu2014crowdsourced}.
\end{enumerate}

\begin{table}[htbp]
\centering
\caption{Comparison of different methods on video tag extraction of the top 10 candidate tags with high-density TSCs.}
\begin{tabular}{|c|c|c|c|c|}
\hline
Method & Prec & Recall & F1-score & MAP\\
\hline
TF-IDF & 0.2674 & 0.5735 & 0.3648 & 0.4224\\
\hline
TX & 0.2427 & 0.5205 & 0.3310 & 0.3696\\
\hline
BTM & 0.2337 & 0.5012	& 0.3188 & 0.3094\\
\hline
GSDPMM & 0.2445 & 0.5094 & 0.3302 & 0.3374\\
\hline
TPTM & 0.2539 & 0.5446 & 0.3463 & 0.3824\\
\hline
SW-IDF (dialogue) & 0.3079 & 0.6602& 0.4210 & 0.4932\\
\hline
SW-IDF (topic center) & 0.3258 & 0.6988 & 0.4444 & 0.5122\\
\hline
\end{tabular}

\label{table:2}
\end{table}

\begin{table}[htbp]
\centering
\caption{Comparison of different methods on video tag extraction of the top 10 candidate tags with low-density TSCs.}
\begin{tabular}{|c|c|c|c|c|}
\hline
Method & Prec & Recall & F1-score & MAP\\
\hline
TF-IDF & 0.3411 & 0.4028 & 0.3694 & 0.3098\\
\hline
TX & 0.3224 & 0.3709 & 0.3450 & 0.3147\\
\hline
BTM & 0.3210 & 0.3662	& 0.3369 & 0.2927\\
\hline
GSDPMM & 0.3440 & 0.4038 & 0.3715 & 0.3202\\
\hline
TPTM & 0.3677 & 0.4334 &  0.3979 & 0.3359 \\
\hline
SW-IDF (dialogue) & 0.3912 & 0.4693 & 0.4267 &0.3623\\
\hline
SW-IDF (topic center)&0.3877 &0.4562&0.4207&0.3522\\
\hline
\end{tabular}
\label{table:3}
\end{table}

For each method, we calculate the precision, recall, MAP (Mean Average Precision) and F1-score of top 10 tagging results at first. Results of high-density and low-density of TSCs are shown in Table \ref{table:2} and Table \ref{table:3}, respectively.

In high-density condition, our topic center-based SW-IDF algorithm achieves optimal results in both F1-score and MAP. It increases the F1-score by 21.82\% and the MAP by 21.26\% compared with the state-of-the-art method TF-IDF in the baselines. In low-density condition, our dialogue-based SW-IDF algorithm achieves optimal results in both F1-score and MAP. It increases the F1-score by 7.24\% and the MAP by 7.86\% compared with the state-of-the-art method TPTM in the baselines. Compare the two algorithms, we find that the dialogue-based algorithm performs better in low-density condition, while topic center-based algorithm performs better in high-density condition, which further proves our assumption in Section \ref{3.2}. 

What is more, when the density of TSCs becomes high, the noises increase. Therefore the result of topic model based methods, BTM, GSDPMM, and TPTM are poor and even worse than classical method TF-IDF. However, TF-IDF only counts the number of words and does not consider the semantic relationship of TSCs, so the result is not as good as our algorithms. Relatively, in low-density comments, the graph is sparse and noises reduce. That is why our algorithms achieve greater improvement in high-density than in low-density.

\begin{table}[h]
\centering
\caption{Comparison of different methods on video tag extraction of the top 5 and top 15 candidate tags}
\begin{tabular}{|c|c|c|c|c|}
\hline
Method & H-Top 5 & H-Top 15 & L-Top 5 & L-top 15\\
   &Prec Recall & Prec Recall &Prec Recall & Prec Recall\\
\hline
TF-IDF & 0.4182 0.4483 & 0.1871 0.5997 & 0.4140 0.2434 & 0.2993 0.5255\\
\hline
TX & 0.3012 0.3234 & 0.1810 0.5831 & 0.3838 0.2250 & 0.2814 0.5071\\
\hline
BTM & 0.2715 0.2924 & 0.1771 0.5692 & 0.3678 0.2158 & 0.2609 0.4602\\
\hline
GSDPMM & 0.2812 0.3013 & 0.1832 0.5930 & 0.4181 0.2486 & 0.3067 0.5390\\
\hline
TPTM & 0.3627 0.3945 & 0.1805 0.5927& 0.4365 0.2662 & 0.3183 0.5624\\
\hline
SW-IDF(d) & 0.4935 0.5362 & 0.2273 0.7241& 0.4654 0.2893 & 0.3556 0.6327\\
\hline
SW-IDF(c) & 0.5300 0.5692 & 0.2345 0.7571 & 0.4518 0.2783 & 0.3410 0.6269\\
\hline
\end{tabular}
\label{table:4}
\end{table}

To further validate our algorithm, we show the precision and recall of top 5 and top 15 candidate tags in Table \ref{table:4}. The results of each algorithm are similar to the performance of Top 10, which prove that our two algorithms have better performance when extracting video tags from time-sync comments in any situation.

\begin{table}[htbp]\footnotesize
\centering
\caption{The top5 results of video tags generated by different algorithms}
\begin{tabular}{|c|c|c|c|c|}
\hline
\label{tags}
Video number &AcFun ac2643295\_1 &AcFun  ac2656362\_6 &AcFun  ac2474006\_1 &AcFun  ac2669229\_1 \\
\hline
Screenshot & \mpage{\includegraphics[width=0.18\linewidth]{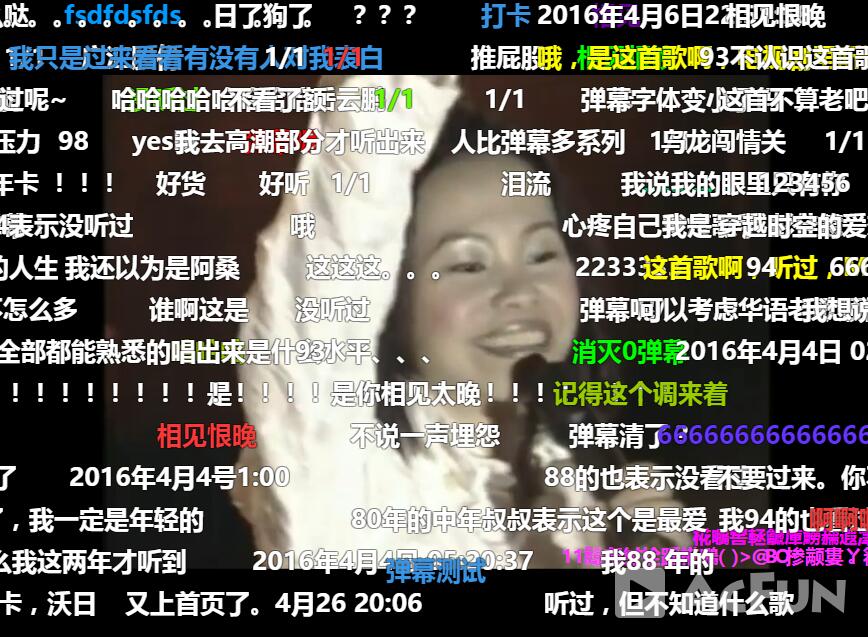}} & \mpage{\includegraphics[width=0.18\linewidth]{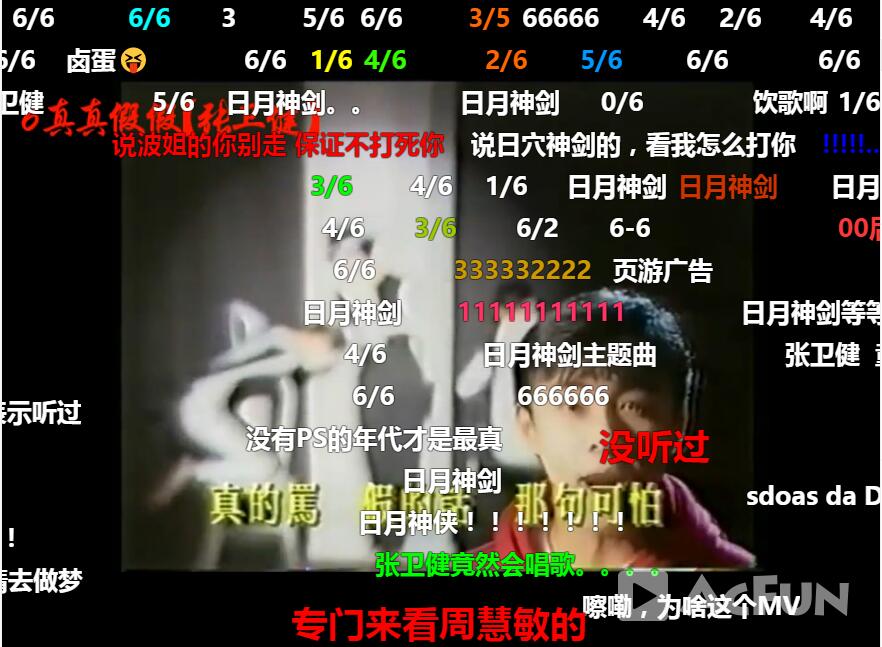}} & \mpage{\includegraphics[width=0.18\linewidth]{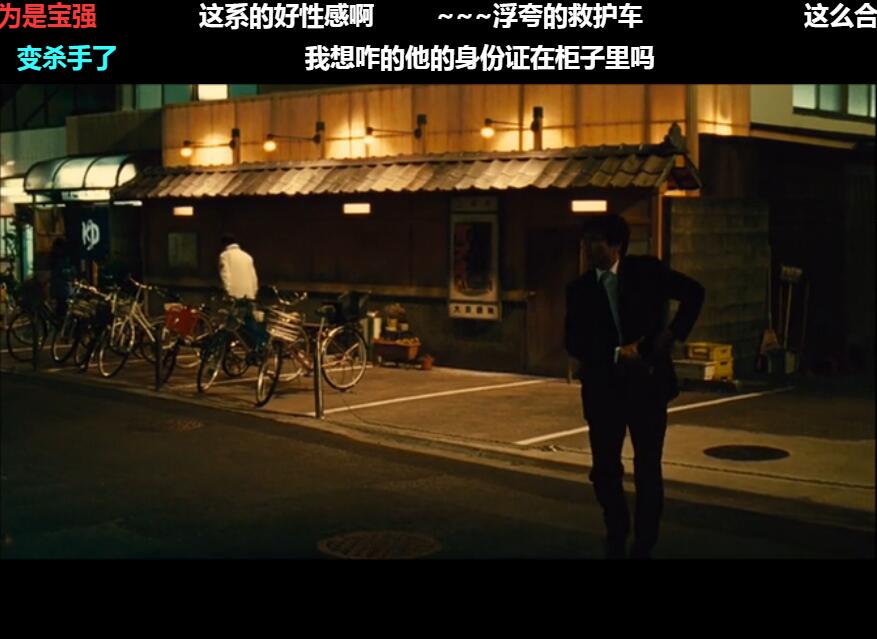}} &  \mpage{\includegraphics[width=0.18\linewidth]{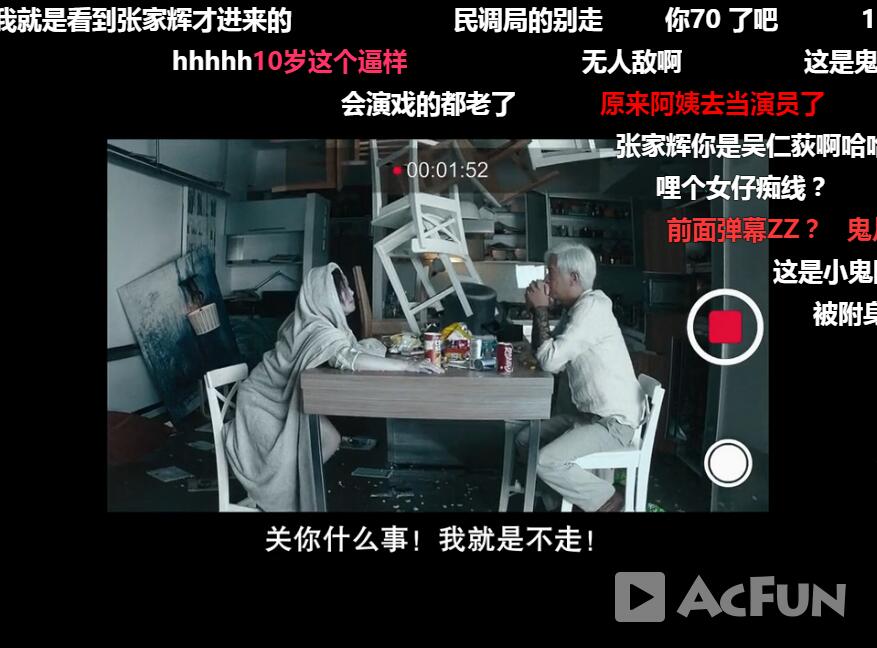}} \\
\hline
Timeline & 0:00:00$\sim$0:01:10 & 0:07:28$\sim$ 0:09:49 & 0:00:00$\sim$1:04:07 & 0:00:00$\sim$0:15:41\\
\hline
Amount & 785& 764& 2933&2460\\
\hline
Density & 672.84& 325.08& 45.78&156.84\\
\hline
TF-IDF & \tabincell{c}{\emph{\textbf{Brief Encounter}} \\ \emph{\textbf{Peng Julia}} \\ \underline{miss} \\ euphonious\\ \emph{\textbf{Wind and Cloud}}}
&\tabincell{c}{ \emph{\textbf{the Twin Swords}}\\  \emph{\textbf{Cheung Wai Kin}}\\  \emph{\textbf{Jen Hsien-chi}}\\  \emph{\textbf{Jimmy Lin}}\\idol}
  & \tabincell{c}{\emph{\textbf{killer}}\\ \emph{\textbf{Ryoko}}\\ ID card\\cell phone\\ \underline{acting skill}}
  &\tabincell{c}{ \emph{\textbf{Cheung Ka Fai}}\\ \underline{alert}\\shock\\  ghost\\  \emph{\textbf{Louis Cheung}}}\\
\hline
TextRank &\tabincell{c}{\emph{\textbf{Brief Encounter}}\\euphonious\\ \underline{our}\\ \emph{\textbf{Peng Julia}}\\ \underline{know}}
&\tabincell{c}{ \emph{\textbf{Jen Hsien-chi}}\\ \emph{\textbf{Cheung Wai Kin}}\\ \underline{hair}\\ \underline{wonderful}\\memory}
 &\tabincell{c}{\emph{\textbf{killer}}\\ ID card\\ actor\\ \emph{\textbf{Japan}}\\ \emph{\textbf{corpse}}}
  &\tabincell{c}{ \emph{\textbf{Cheung Ka Fai}}\\ \underline{ alert}\\movie\\ terror\\ \underline{feel}} \\
\hline
BTM &\tabincell{c}{\emph{\textbf{Brief Encounter}}\\euphonious \\ \underline{know}\\ \underline{brave} \\ childhood}
 &\tabincell{c}{ \emph{\textbf{Jen Hsien-chi}}\\ \underline{hair}\\ \underline{wonderful}\\memery\\  \emph{\textbf{Cantonese}}}
 & \tabincell{c}{\underline{fierce}\\ \underline{perform}\\ \underline{model}\\employer\\ \emph{\textbf{corpse}}}
 &\tabincell{c}{\underline{update}\\ \underline{hobbies}\\ movie\\ \emph{\textbf{Cheung Ka Fai}}\\ \underline{forget}} \\
\hline
GSDPMM & \tabincell{c}{euphonious \\ \underline{follow}\\ \underline{myself}\\ \underline{brave}\\ \emph{\textbf{Wind and Cloud}}}
 & \tabincell{c}{ \emph{\textbf{Jen Hsien-chi}}\\ \underline{hair}\\  \emph{\textbf{Cantonese}}\\ \emph{\textbf{Jimmy Lin}}\\ \underline{love}}
  & \tabincell{c}{\underline{New Year}\\ \emph{\textbf{killer}}\\ ID card\\ \underline{chimney}\\bathhouse}
   &\tabincell{c}{\underline{powerful}\\ \underline{alerf}\\ movie\\ ghost\\ fear}\\
\hline
TPTM& \tabincell{c}{ \emph{\textbf{Brief Encounter }}\\ \underline{stuck}\\ \emph{\textbf{Peng Julia}}\\ euphonious \\ childhood}
& \tabincell{c}{ \emph{\textbf{Cheung Wai Kin}}\\  memory\\ \emph{\textbf{the Twin Swords}}\\ \emph{\textbf{Jimmy Lin}}\\ \underline{hair}}
 & \tabincell{c}{ \emph{\textbf{corpse}}\\ cell phone\\ \emph{\textbf{killer}}\\  \emph{\textbf{Japan}}\\  actor}
 &\tabincell{c}{ \emph{\textbf{Cheung Ka Fai}} \\ movie\\ \underline{forget}\\  ghost\\  \emph{\textbf{dracula movie}}}\\
\hline
SW-IDF(d) & \tabincell{c}{\emph{\textbf{Brief Encounter}} \\ \emph{\textbf{Peng Julia}}\\ euphonious\\ \underline{express}\\ \emph{\textbf{Wind and Cloud}}}
& \tabincell{c}{ \emph{\textbf{the Twin Swords}}\\ \emph{\textbf{Cheung Wai Kin}}\\ \emph{\textbf{Jen Hsien-chi}}\\ \emph{\textbf{Jimmy Lin}}\\idol}
& \tabincell{c}{ \emph{\textbf{killer}}\\ ID card\\  \emph{\textbf{Kreisler}}\\  \emph{\textbf{Japan}}\\ cell phone}
&\tabincell{c}{ \emph{\textbf{Lawrence}}\\ \emph{\textbf{Cheung Ka Fai}}\\ \emph{\textbf{Louis Cheung}}\\ \underline{alert}\\shock\\}\\
\hline
SW-IDF(c) & \tabincell{c}{ \emph{\textbf{Brief Encounter }}\\ \emph{\textbf{Peng Julia}} \\ \emph{\textbf{Wind and Cloud}}\\ euphonious \\
 \emph{\textbf{theme song}}}
& \tabincell{c}{ \emph{\textbf{Cheung Wai Kin}}\\  \emph{\textbf{the Twin Swords}}\\ \emph{\textbf{Jimmy Lin}}\\memory\\ \emph{\textbf{Jen Hsien-chi}}}
 & \tabincell{c}{ \emph{\textbf{killer}}\\ ID card\\  \emph{\textbf{Japan}}\\  \emph{\textbf{Ryoko}}\\ cell phone}
 &\tabincell{c}{ \emph{\textbf{Cheung Ka Fai}} \\ shock\\  ghost\\  \emph{\textbf{dracula movie}}\\ \emph{\textbf{Kuo Tsai-chieh}}}\\
\hline
\end{tabular}
\end{table}

Finally, we show the Top 5 of video tags generated by the algorithms above in Table \ref{tags}. The \emph{\textbf{Bold italic words}} indicate the good tags (the tags that all three volunteers voted), while the \underline{underline words} indicate the bad tags ((the tags that less than two volunteers voted)). The results show that the SW-IDF (Topic Center) and SW-IDF(dialogue) have more good tags and less bad tags than other algorithms, which intuitively demonstrates the superiority of our algorithms.

\section{Conclusion}
\label{sec:5}
In this paper, we proposed a novel video tag extraction algorithm to acquire video tags for time-sync videos. To deal with the features of time-sync comments, SW-IDF was designed to cluster comments into semantic association graph by taking advantage of their semantic similarities and timestamps. In this way, the noises could be differentiated from the meaningful comments, and thus be effectively eliminated. Finally, video tags were well recognized and extracted in an unsupervised way. Extensive experiments on real-world dataset proved that our algorithm could effectively extract video tags with a significant improvement of precision and recall compared with several baselines, which obviously validates the potential of our algorithm on tag extraction, as well as tackling with the features of time-sync comments.

\begin{acks}

This work is supported by Chinese National Research Fund (NSFC) Key Project No. 61532013 and No. 61872239. NSFC Project No. 61872195 and No. 61702330. FDCT/0007/2018/A1, DCT-MoST Joint-project No. (025/2015/AMJ), University of Macau Grant Nos: MYRG2018-00237-RTO, CPG2018-00032-FST and SRG2018-00111-FST of SAR Macau, China.

\end{acks}

\bibliographystyle{ACM-Reference-Format}
\bibliography{sample-acmsmall}


\begin{thebibliography}{62}


\ifx \showCODEN    \undefined \def \showCODEN     #1{\unskip}     \fi
\ifx \showDOI      \undefined \def \showDOI       #1{#1}\fi
\ifx \showISBNx    \undefined \def \showISBNx     #1{\unskip}     \fi
\ifx \showISBNxiii \undefined \def \showISBNxiii  #1{\unskip}     \fi
\ifx \showISSN     \undefined \def \showISSN      #1{\unskip}     \fi
\ifx \showLCCN     \undefined \def \showLCCN      #1{\unskip}     \fi
\ifx \shownote     \undefined \def \shownote      #1{#1}          \fi
\ifx \showarticletitle \undefined \def \showarticletitle #1{#1}   \fi
\ifx \showURL      \undefined \def \showURL       {\relax}        \fi
\providecommand\bibfield[2]{#2}
\providecommand\bibinfo[2]{#2}
\providecommand\natexlab[1]{#1}
\providecommand\showeprint[2][]{arXiv:#2}

\bibitem[\protect\citeauthoryear{Alhabashneh, Iqbal, Doctor, and
  James}{Alhabashneh et~al\mbox{.}}{2017}]%
        {alhabashneh2017fuzzy}
\bibfield{author}{\bibinfo{person}{Obada Alhabashneh}, \bibinfo{person}{Rahat
  Iqbal}, \bibinfo{person}{Faiyaz Doctor}, {and} \bibinfo{person}{Anne James}.}
  \bibinfo{year}{2017}\natexlab{}.
\newblock \showarticletitle{Fuzzy rule based profiling approach for enterprise
  information seeking and retrieval}.
\newblock \bibinfo{journal}{\emph{Information Sciences}}  \bibinfo{volume}{394}
  (\bibinfo{year}{2017}), \bibinfo{pages}{18--37}.
\newblock


\bibitem[\protect\citeauthoryear{Alstrup, Brodal, and Rauhe}{Alstrup
  et~al\mbox{.}}{2000}]%
        {alstrup2000new}
\bibfield{author}{\bibinfo{person}{Stephen Alstrup},
  \bibinfo{person}{G~Stolting Brodal}, {and} \bibinfo{person}{Theis Rauhe}.}
  \bibinfo{year}{2000}\natexlab{}.
\newblock \showarticletitle{New data structures for orthogonal range
  searching}. In \bibinfo{booktitle}{\emph{Foundations of Computer Science,
  2000. Proceedings. 41st Annual Symposium on}}. IEEE,
  \bibinfo{pages}{198--207}.
\newblock


\bibitem[\protect\citeauthoryear{Bae, Halperin, West, Rosvall, and Howe}{Bae
  et~al\mbox{.}}{2017}]%
        {bae2017scalable}
\bibfield{author}{\bibinfo{person}{Seung-Hee Bae}, \bibinfo{person}{Daniel
  Halperin}, \bibinfo{person}{Jevin~D West}, \bibinfo{person}{Martin Rosvall},
  {and} \bibinfo{person}{Bill Howe}.} \bibinfo{year}{2017}\natexlab{}.
\newblock \showarticletitle{Scalable and efficient flow-based community
  detection for large-scale graph analysis}.
\newblock \bibinfo{journal}{\emph{ACM Transactions on Knowledge Discovery from
  Data (TKDD)}} \bibinfo{volume}{11}, \bibinfo{number}{3}
  (\bibinfo{year}{2017}), \bibinfo{pages}{32}.
\newblock


\bibitem[\protect\citeauthoryear{Bentley}{Bentley}{1975}]%
        {bentley1975multidimensional}
\bibfield{author}{\bibinfo{person}{Jon~Louis Bentley}.}
  \bibinfo{year}{1975}\natexlab{}.
\newblock \showarticletitle{Multidimensional binary search trees used for
  associative searching}.
\newblock \bibinfo{journal}{\emph{Commun. ACM}} \bibinfo{volume}{18},
  \bibinfo{number}{9} (\bibinfo{year}{1975}), \bibinfo{pages}{509--517}.
\newblock


\bibitem[\protect\citeauthoryear{Bentley}{Bentley}{1990}]%
        {bentley1990k}
\bibfield{author}{\bibinfo{person}{Jon~Louis Bentley}.}
  \bibinfo{year}{1990}\natexlab{}.
\newblock \showarticletitle{K-d trees for semidynamic point sets}. In
  \bibinfo{booktitle}{\emph{Proceedings of the sixth annual symposium on
  Computational geometry}}. ACM, \bibinfo{pages}{187--197}.
\newblock


\bibitem[\protect\citeauthoryear{Blei, Ng, and Jordan}{Blei
  et~al\mbox{.}}{2003}]%
        {blei2003latent}
\bibfield{author}{\bibinfo{person}{David~M Blei}, \bibinfo{person}{Andrew~Y
  Ng}, {and} \bibinfo{person}{Michael~I Jordan}.}
  \bibinfo{year}{2003}\natexlab{}.
\newblock \showarticletitle{Latent dirichlet allocation}.
\newblock \bibinfo{journal}{\emph{Journal of Machine Learning Research}}
  \bibinfo{volume}{3} (\bibinfo{year}{2003}), \bibinfo{pages}{993--1022}.
\newblock


\bibitem[\protect\citeauthoryear{Bordino, Castillo, Donato, and Gionis}{Bordino
  et~al\mbox{.}}{2010}]%
        {bordino2010query}
\bibfield{author}{\bibinfo{person}{Ilaria Bordino}, \bibinfo{person}{Carlos
  Castillo}, \bibinfo{person}{Debora Donato}, {and} \bibinfo{person}{Aristides
  Gionis}.} \bibinfo{year}{2010}\natexlab{}.
\newblock \showarticletitle{Query similarity by projecting the query-flow
  graph}. In \bibinfo{booktitle}{\emph{Proceedings of the 33rd international
  ACM SIGIR conference on Research and development in information retrieval}}.
  ACM, \bibinfo{pages}{515--522}.
\newblock


\bibitem[\protect\citeauthoryear{Chakraborty, Srinivasan, Ganguly, Mukherjee,
  and Bhowmick}{Chakraborty et~al\mbox{.}}{2016}]%
        {chakraborty2016permanence}
\bibfield{author}{\bibinfo{person}{Tanmoy Chakraborty}, \bibinfo{person}{Sriram
  Srinivasan}, \bibinfo{person}{Niloy Ganguly}, \bibinfo{person}{Animesh
  Mukherjee}, {and} \bibinfo{person}{Sanjukta Bhowmick}.}
  \bibinfo{year}{2016}\natexlab{}.
\newblock \showarticletitle{Permanence and community structure in complex
  networks}.
\newblock \bibinfo{journal}{\emph{ACM Transactions on Knowledge Discovery from
  Data (TKDD)}} \bibinfo{volume}{11}, \bibinfo{number}{2}
  (\bibinfo{year}{2016}), \bibinfo{pages}{14}.
\newblock


\bibitem[\protect\citeauthoryear{Chen, Chen, Jin, and Hauptmann}{Chen
  et~al\mbox{.}}{2017a}]%
        {chen2017video}
\bibfield{author}{\bibinfo{person}{Shizhe Chen}, \bibinfo{person}{Jia Chen},
  \bibinfo{person}{Qin Jin}, {and} \bibinfo{person}{Alexander Hauptmann}.}
  \bibinfo{year}{2017}\natexlab{a}.
\newblock \showarticletitle{Video captioning with guidance of multimodal latent
  topics}. In \bibinfo{booktitle}{\emph{Proceedings of the 2017 ACM on
  Multimedia Conference}}. ACM, \bibinfo{pages}{1838--1846}.
\newblock


\bibitem[\protect\citeauthoryear{Chen, Zhang, Ai, Xu, Yan, and Qin}{Chen
  et~al\mbox{.}}{2017b}]%
        {chen2017personalized}
\bibfield{author}{\bibinfo{person}{Xu Chen}, \bibinfo{person}{Yongfeng Zhang},
  \bibinfo{person}{Qingyao Ai}, \bibinfo{person}{Hongteng Xu},
  \bibinfo{person}{Junchi Yan}, {and} \bibinfo{person}{Zheng Qin}.}
  \bibinfo{year}{2017}\natexlab{b}.
\newblock \showarticletitle{Personalized key frame recommendation}. In
  \bibinfo{booktitle}{\emph{Proceedings of the 40th International ACM SIGIR
  Conference on Research and Development in Information Retrieval}}. ACM,
  \bibinfo{pages}{315--324}.
\newblock


\bibitem[\protect\citeauthoryear{Dong and Dong}{Dong and Dong}{2003}]%
        {dong2003hownet}
\bibfield{author}{\bibinfo{person}{Zhendong Dong} {and} \bibinfo{person}{Qiang
  Dong}.} \bibinfo{year}{2003}\natexlab{}.
\newblock \showarticletitle{HowNet-a hybrid language and knowledge resource}.
  In \bibinfo{booktitle}{\emph{Natural Language Processing and Knowledge
  Engineering, 2003. Proceedings. 2003 International Conference on}}. IEEE,
  \bibinfo{pages}{820--824}.
\newblock


\bibitem[\protect\citeauthoryear{Fortunato}{Fortunato}{2010}]%
        {fortunato2010community}
\bibfield{author}{\bibinfo{person}{Santo Fortunato}.}
  \bibinfo{year}{2010}\natexlab{}.
\newblock \showarticletitle{Community detection in graphs}.
\newblock \bibinfo{journal}{\emph{Physics reports}} \bibinfo{volume}{486},
  \bibinfo{number}{3} (\bibinfo{year}{2010}), \bibinfo{pages}{75--174}.
\newblock


\bibitem[\protect\citeauthoryear{Friedman, Bentley, and Finkel}{Friedman
  et~al\mbox{.}}{1977}]%
        {friedman1977algorithm}
\bibfield{author}{\bibinfo{person}{Jerome~H Friedman},
  \bibinfo{person}{Jon~Louis Bentley}, {and} \bibinfo{person}{Raphael~Ari
  Finkel}.} \bibinfo{year}{1977}\natexlab{}.
\newblock \showarticletitle{An algorithm for finding best matches in
  logarithmic expected time}.
\newblock \bibinfo{journal}{\emph{ACM Transactions on Mathematical Software
  (TOMS)}} \bibinfo{volume}{3}, \bibinfo{number}{3} (\bibinfo{year}{1977}),
  \bibinfo{pages}{209--226}.
\newblock


\bibitem[\protect\citeauthoryear{Gu, Wang, Liu, Guo, and Liu}{Gu
  et~al\mbox{.}}{2017}]%
        {gu2017reliable}
\bibfield{author}{\bibinfo{person}{Liqiu Gu}, \bibinfo{person}{Kun Wang},
  \bibinfo{person}{Xiulong Liu}, \bibinfo{person}{Song Guo}, {and}
  \bibinfo{person}{Bo Liu}.} \bibinfo{year}{2017}\natexlab{}.
\newblock \showarticletitle{A reliable task assignment strategy for spatial
  crowdsourcing in big data environment}. In \bibinfo{booktitle}{\emph{2017
  IEEE International Conference on Communications (ICC)}}. IEEE,
  \bibinfo{pages}{1--6}.
\newblock


\bibitem[\protect\citeauthoryear{Guo, Cheng, Xu, and Zhu}{Guo
  et~al\mbox{.}}{2011}]%
        {guo2011intent}
\bibfield{author}{\bibinfo{person}{Jiafeng Guo}, \bibinfo{person}{Xueqi Cheng},
  \bibinfo{person}{Gu Xu}, {and} \bibinfo{person}{Xiaofei Zhu}.}
  \bibinfo{year}{2011}\natexlab{}.
\newblock \showarticletitle{Intent-aware query similarity}. In
  \bibinfo{booktitle}{\emph{Proceedings of the 20th ACM international
  conference on Information and knowledge management}}. ACM,
  \bibinfo{pages}{259--268}.
\newblock


\bibitem[\protect\citeauthoryear{He, Gimpel, and Lin}{He et~al\mbox{.}}{2015}]%
        {he2015multi}
\bibfield{author}{\bibinfo{person}{Hua He}, \bibinfo{person}{Kevin Gimpel},
  {and} \bibinfo{person}{Jimmy~J Lin}.} \bibinfo{year}{2015}\natexlab{}.
\newblock \showarticletitle{Multi-Perspective Sentence Similarity Modeling with
  Convolutional Neural Networks.}. In \bibinfo{booktitle}{\emph{EMNLP}}.
  \bibinfo{pages}{1576--1586}.
\newblock


\bibitem[\protect\citeauthoryear{He, Ge, Wu, Chen, and Tan}{He
  et~al\mbox{.}}{2016}]%
        {he2016predicting}
\bibfield{author}{\bibinfo{person}{Ming He}, \bibinfo{person}{Yong Ge},
  \bibinfo{person}{Le Wu}, \bibinfo{person}{Enhong Chen}, {and}
  \bibinfo{person}{Chang Tan}.} \bibinfo{year}{2016}\natexlab{}.
\newblock \showarticletitle{Predicting the Popularity of DanMu-enabled Videos:
  A Multi-factor View}. In \bibinfo{booktitle}{\emph{Proceedings of
  International Conference on Database Systems for Advanced Applications}}.
  Springer, \bibinfo{pages}{351--366}.
\newblock


\bibitem[\protect\citeauthoryear{Huang, Li, Zhang, Zhang, Chen, and Zhai}{Huang
  et~al\mbox{.}}{2017}]%
        {huang2017overlapping}
\bibfield{author}{\bibinfo{person}{Faliang Huang}, \bibinfo{person}{Xuelong
  Li}, \bibinfo{person}{Shichao Zhang}, \bibinfo{person}{Jilian Zhang},
  \bibinfo{person}{Jinhui Chen}, {and} \bibinfo{person}{Zhinian Zhai}.}
  \bibinfo{year}{2017}\natexlab{}.
\newblock \showarticletitle{Overlapping community detection for multimedia
  social networks}.
\newblock \bibinfo{journal}{\emph{IEEE Transactions on Multimedia}}
  \bibinfo{volume}{19}, \bibinfo{number}{8} (\bibinfo{year}{2017}),
  \bibinfo{pages}{1881--1893}.
\newblock


\bibitem[\protect\citeauthoryear{Hussein and Piccardi}{Hussein and
  Piccardi}{2017}]%
        {hussein2017v}
\bibfield{author}{\bibinfo{person}{Fairouz Hussein} {and}
  \bibinfo{person}{Massimo Piccardi}.} \bibinfo{year}{2017}\natexlab{}.
\newblock \showarticletitle{V-JAUNE: A Framework for Joint Action Recognition
  and Video Summarization}.
\newblock \bibinfo{journal}{\emph{ACM Transactions on Multimedia Computing,
  Communications, and Applications (TOMM)}} \bibinfo{volume}{13},
  \bibinfo{number}{2} (\bibinfo{year}{2017}), \bibinfo{pages}{20}.
\newblock


\bibitem[\protect\citeauthoryear{Hyung, Park, and Lee}{Hyung
  et~al\mbox{.}}{2017}]%
        {hyung2017utilizing}
\bibfield{author}{\bibinfo{person}{Ziwon Hyung}, \bibinfo{person}{Joon-Sang
  Park}, {and} \bibinfo{person}{Kyogu Lee}.} \bibinfo{year}{2017}\natexlab{}.
\newblock \showarticletitle{Utilizing context-relevant keywords extracted from
  a large collection of user-generated documents for music discovery}.
\newblock \bibinfo{journal}{\emph{Information Processing \& Management}}
  \bibinfo{volume}{53}, \bibinfo{number}{5} (\bibinfo{year}{2017}),
  \bibinfo{pages}{1185--1200}.
\newblock


\bibitem[\protect\citeauthoryear{Iacobacci, Pilehvar, and Navigli}{Iacobacci
  et~al\mbox{.}}{2015}]%
        {iacobacci2015sensembed}
\bibfield{author}{\bibinfo{person}{Ignacio Iacobacci},
  \bibinfo{person}{Mohammad~Taher Pilehvar}, {and} \bibinfo{person}{Roberto
  Navigli}.} \bibinfo{year}{2015}\natexlab{}.
\newblock \showarticletitle{Sensembed: Learning sense embeddings for word and
  relational similarity}. In \bibinfo{booktitle}{\emph{Proceedings of the 53rd
  Annual Meeting of the Association for Computational Linguistics and the 7th
  International Joint Conference on Natural Language Processing}},
  Vol.~\bibinfo{volume}{1}. \bibinfo{pages}{95--105}.
\newblock


\bibitem[\protect\citeauthoryear{Kenter and De~Rijke}{Kenter and
  De~Rijke}{2015}]%
        {kenter2015short}
\bibfield{author}{\bibinfo{person}{Tom Kenter} {and} \bibinfo{person}{Maarten
  De~Rijke}.} \bibinfo{year}{2015}\natexlab{}.
\newblock \showarticletitle{Short text similarity with word embeddings}. In
  \bibinfo{booktitle}{\emph{Proceedings of the 24th ACM international on
  conference on information and knowledge management}}. ACM,
  \bibinfo{pages}{1411--1420}.
\newblock


\bibitem[\protect\citeauthoryear{Kusner, Sun, Kolkin, and Weinberger}{Kusner
  et~al\mbox{.}}{2015}]%
        {kusner2015word}
\bibfield{author}{\bibinfo{person}{Matt~J Kusner}, \bibinfo{person}{Yu Sun},
  \bibinfo{person}{Nicholas~I Kolkin}, {and} \bibinfo{person}{Kilian~Q
  Weinberger}.} \bibinfo{year}{2015}\natexlab{}.
\newblock \showarticletitle{From word embeddings to document distances}. In
  \bibinfo{booktitle}{\emph{Proceedings of the 32nd International Conference on
  Machine Learning (ICML 2015)}}. \bibinfo{pages}{957--966}.
\newblock


\bibitem[\protect\citeauthoryear{Lancichinetti and Fortunato}{Lancichinetti and
  Fortunato}{2009}]%
        {lancichinetti2009community}
\bibfield{author}{\bibinfo{person}{Andrea Lancichinetti} {and}
  \bibinfo{person}{Santo Fortunato}.} \bibinfo{year}{2009}\natexlab{}.
\newblock \showarticletitle{Community detection algorithms: a comparative
  analysis}.
\newblock \bibinfo{journal}{\emph{Physical review E}} \bibinfo{volume}{80},
  \bibinfo{number}{5} (\bibinfo{year}{2009}), \bibinfo{pages}{056117}.
\newblock


\bibitem[\protect\citeauthoryear{Lancichinetti, Fortunato, and
  Radicchi}{Lancichinetti et~al\mbox{.}}{2008}]%
        {lancichinetti2008benchmark}
\bibfield{author}{\bibinfo{person}{Andrea Lancichinetti},
  \bibinfo{person}{Santo Fortunato}, {and} \bibinfo{person}{Filippo Radicchi}.}
  \bibinfo{year}{2008}\natexlab{}.
\newblock \showarticletitle{Benchmark graphs for testing community detection
  algorithms}.
\newblock \bibinfo{journal}{\emph{Physical review E}} \bibinfo{volume}{78},
  \bibinfo{number}{4} (\bibinfo{year}{2008}), \bibinfo{pages}{046110}.
\newblock


\bibitem[\protect\citeauthoryear{Lee and Wong}{Lee and Wong}{1977}]%
        {lee1977worst}
\bibfield{author}{\bibinfo{person}{Der-Tsai Lee} {and} \bibinfo{person}{CK
  Wong}.} \bibinfo{year}{1977}\natexlab{}.
\newblock \showarticletitle{Worst-case analysis for region and partial region
  searches in multidimensional binary search trees and balanced quad trees}.
\newblock \bibinfo{journal}{\emph{Acta Informatica}} \bibinfo{volume}{9},
  \bibinfo{number}{1} (\bibinfo{year}{1977}), \bibinfo{pages}{23--29}.
\newblock


\bibitem[\protect\citeauthoryear{Levy and Goldberg}{Levy and Goldberg}{2014a}]%
        {levy2014linguistic}
\bibfield{author}{\bibinfo{person}{Omer Levy} {and} \bibinfo{person}{Yoav
  Goldberg}.} \bibinfo{year}{2014}\natexlab{a}.
\newblock \showarticletitle{Linguistic regularities in sparse and explicit word
  representations}. In \bibinfo{booktitle}{\emph{Proceedings of the eighteenth
  conference on computational natural language learning}}.
  \bibinfo{pages}{171--180}.
\newblock


\bibitem[\protect\citeauthoryear{Levy and Goldberg}{Levy and Goldberg}{2014b}]%
        {levy2014neural}
\bibfield{author}{\bibinfo{person}{Omer Levy} {and} \bibinfo{person}{Yoav
  Goldberg}.} \bibinfo{year}{2014}\natexlab{b}.
\newblock \showarticletitle{Neural word embedding as implicit matrix
  factorization}. In \bibinfo{booktitle}{\emph{Advances in neural information
  processing systems}}. \bibinfo{pages}{2177--2185}.
\newblock


\bibitem[\protect\citeauthoryear{Levy, Goldberg, and Dagan}{Levy
  et~al\mbox{.}}{2015}]%
        {levy2015improving}
\bibfield{author}{\bibinfo{person}{Omer Levy}, \bibinfo{person}{Yoav Goldberg},
  {and} \bibinfo{person}{Ido Dagan}.} \bibinfo{year}{2015}\natexlab{}.
\newblock \showarticletitle{Improving distributional similarity with lessons
  learned from word embeddings}.
\newblock \bibinfo{journal}{\emph{Transactions of the Association for
  Computational Linguistics}}  \bibinfo{volume}{3} (\bibinfo{year}{2015}),
  \bibinfo{pages}{211--225}.
\newblock


\bibitem[\protect\citeauthoryear{Li, Zhao, Hu, Li, Liu, and Du}{Li
  et~al\mbox{.}}{2018b}]%
        {li2018analogical}
\bibfield{author}{\bibinfo{person}{Shen Li}, \bibinfo{person}{Zhe Zhao},
  \bibinfo{person}{Renfen Hu}, \bibinfo{person}{Wensi Li}, \bibinfo{person}{Tao
  Liu}, {and} \bibinfo{person}{Xiaoyong Du}.} \bibinfo{year}{2018}\natexlab{b}.
\newblock \showarticletitle{Analogical Reasoning on Chinese Morphological and
  Semantic Relations}. In \bibinfo{booktitle}{\emph{Proceedings of the 56th
  Annual Meeting of the Association for Computational Linguistics (Volume 2:
  Short Papers}}. \bibinfo{publisher}{ACL}, \bibinfo{pages}{138--143}.
\newblock


\bibitem[\protect\citeauthoryear{Li, He, Kloster, Bindel, and Hopcroft}{Li
  et~al\mbox{.}}{2018a}]%
        {li2018local}
\bibfield{author}{\bibinfo{person}{Yixuan Li}, \bibinfo{person}{Kun He},
  \bibinfo{person}{Kyle Kloster}, \bibinfo{person}{David Bindel}, {and}
  \bibinfo{person}{John Hopcroft}.} \bibinfo{year}{2018}\natexlab{a}.
\newblock \showarticletitle{Local Spectral Clustering for Overlapping Community
  Detection}.
\newblock \bibinfo{journal}{\emph{ACM Transactions on Knowledge Discovery from
  Data (TKDD)}} \bibinfo{volume}{12}, \bibinfo{number}{2}
  (\bibinfo{year}{2018}), \bibinfo{pages}{17}.
\newblock


\bibitem[\protect\citeauthoryear{Liao, Xian, Yang, Zhao, Zhang, and Li}{Liao
  et~al\mbox{.}}{2018}]%
        {liao2018tscset}
\bibfield{author}{\bibinfo{person}{Zhenyu Liao}, \bibinfo{person}{Yikun Xian},
  \bibinfo{person}{Xiao Yang}, \bibinfo{person}{Qinpei Zhao},
  \bibinfo{person}{Chenxi Zhang}, {and} \bibinfo{person}{Jiangfeng Li}.}
  \bibinfo{year}{2018}\natexlab{}.
\newblock \showarticletitle{TSCSet: A Crowdsourced Time-Sync Comment Dataset
  for Exploration of User Experience Improvement}. In
  \bibinfo{booktitle}{\emph{23rd International Conference on Intelligent User
  Interfaces}}. ACM, \bibinfo{pages}{641--652}.
\newblock


\bibitem[\protect\citeauthoryear{Lv, Xu, Chen, Liu, and Zheng}{Lv
  et~al\mbox{.}}{2016}]%
        {lv2016reading}
\bibfield{author}{\bibinfo{person}{Guangyi Lv}, \bibinfo{person}{Tong Xu},
  \bibinfo{person}{Enhong Chen}, \bibinfo{person}{Qi Liu}, {and}
  \bibinfo{person}{Yi Zheng}.} \bibinfo{year}{2016}\natexlab{}.
\newblock \showarticletitle{Reading the Videos: Temporal Labeling for
  Crowdsourced Time-Sync Videos Based on Semantic Embedding}. In
  \bibinfo{booktitle}{\emph{Proceedings of the 30th AAAI Conference on
  Artificial Intelligence}}.
\newblock


\bibitem[\protect\citeauthoryear{Mihalcea and Tarau}{Mihalcea and
  Tarau}{2004}]%
        {mihalcea2004textrank}
\bibfield{author}{\bibinfo{person}{Rada Mihalcea} {and} \bibinfo{person}{Paul
  Tarau}.} \bibinfo{year}{2004}\natexlab{}.
\newblock \showarticletitle{TextRank: Bringing order into texts}. In
  \bibinfo{booktitle}{\emph{Proceedings of the Conference on Empirical Methods
  in Natural Language Processing}}. Association for Computational Linguistics,
  \bibinfo{pages}{8--15}.
\newblock


\bibitem[\protect\citeauthoryear{Mikolov, Sutskever, Chen, Corrado, and
  Dean}{Mikolov et~al\mbox{.}}{2013}]%
        {mikolov2013distributed}
\bibfield{author}{\bibinfo{person}{Tomas Mikolov}, \bibinfo{person}{Ilya
  Sutskever}, \bibinfo{person}{Kai Chen}, \bibinfo{person}{Greg~S Corrado},
  {and} \bibinfo{person}{Jeff Dean}.} \bibinfo{year}{2013}\natexlab{}.
\newblock \showarticletitle{Distributed representations of words and phrases
  and their compositionality}. In \bibinfo{booktitle}{\emph{Advances in neural
  information processing systems}}. \bibinfo{pages}{3111--3119}.
\newblock


\bibitem[\protect\citeauthoryear{Mueller and Thyagarajan}{Mueller and
  Thyagarajan}{2016}]%
        {mueller2016siamese}
\bibfield{author}{\bibinfo{person}{Jonas Mueller} {and} \bibinfo{person}{Aditya
  Thyagarajan}.} \bibinfo{year}{2016}\natexlab{}.
\newblock \showarticletitle{Siamese Recurrent Architectures for Learning
  Sentence Similarity.}. In \bibinfo{booktitle}{\emph{AAAI}}.
  \bibinfo{pages}{2786--2792}.
\newblock


\bibitem[\protect\citeauthoryear{Murtagh and Contreras}{Murtagh and
  Contreras}{2012}]%
        {murtagh2012algorithms}
\bibfield{author}{\bibinfo{person}{Fionn Murtagh} {and} \bibinfo{person}{Pedro
  Contreras}.} \bibinfo{year}{2012}\natexlab{}.
\newblock \showarticletitle{Algorithms for hierarchical clustering: an
  overview}.
\newblock \bibinfo{journal}{\emph{Wiley Interdisciplinary Reviews: Data Mining
  and Knowledge Discovery}} \bibinfo{volume}{2}, \bibinfo{number}{1}
  (\bibinfo{year}{2012}), \bibinfo{pages}{86--97}.
\newblock


\bibitem[\protect\citeauthoryear{Murtagh, Downs, and Contreras}{Murtagh
  et~al\mbox{.}}{2008}]%
        {murtagh2008hierarchical}
\bibfield{author}{\bibinfo{person}{Fionn Murtagh}, \bibinfo{person}{Geoff
  Downs}, {and} \bibinfo{person}{Pedro Contreras}.}
  \bibinfo{year}{2008}\natexlab{}.
\newblock \showarticletitle{Hierarchical clustering of massive, high
  dimensional data sets by exploiting ultrametric embedding}.
\newblock \bibinfo{journal}{\emph{SIAM Journal on Scientific Computing}}
  \bibinfo{volume}{30}, \bibinfo{number}{2} (\bibinfo{year}{2008}),
  \bibinfo{pages}{707--730}.
\newblock


\bibitem[\protect\citeauthoryear{Murtagh and Legendre}{Murtagh and
  Legendre}{2014}]%
        {murtagh2014ward}
\bibfield{author}{\bibinfo{person}{Fionn Murtagh} {and} \bibinfo{person}{Pierre
  Legendre}.} \bibinfo{year}{2014}\natexlab{}.
\newblock \showarticletitle{Ward's hierarchical agglomerative clustering
  method: which algorithms implement ward's criterion?}
\newblock \bibinfo{journal}{\emph{Journal of Classification}}
  \bibinfo{volume}{31}, \bibinfo{number}{3} (\bibinfo{year}{2014}),
  \bibinfo{pages}{274--295}.
\newblock


\bibitem[\protect\citeauthoryear{Newman and Girvan}{Newman and Girvan}{2004}]%
        {newman2004finding}
\bibfield{author}{\bibinfo{person}{Mark~EJ Newman} {and}
  \bibinfo{person}{Michelle Girvan}.} \bibinfo{year}{2004}\natexlab{}.
\newblock \showarticletitle{Finding and evaluating community structure in
  networks}.
\newblock \bibinfo{journal}{\emph{Physical review E}} \bibinfo{volume}{69},
  \bibinfo{number}{2} (\bibinfo{year}{2004}), \bibinfo{pages}{026113}.
\newblock


\bibitem[\protect\citeauthoryear{Pandove, Goel, and Rani}{Pandove
  et~al\mbox{.}}{2018}]%
        {pandove2018systematic}
\bibfield{author}{\bibinfo{person}{Divya Pandove}, \bibinfo{person}{Shivan
  Goel}, {and} \bibinfo{person}{Rinkl Rani}.} \bibinfo{year}{2018}\natexlab{}.
\newblock \showarticletitle{Systematic review of clustering high-dimensional
  and large datasets}.
\newblock \bibinfo{journal}{\emph{ACM Transactions on Knowledge Discovery from
  Data (TKDD)}} \bibinfo{volume}{12}, \bibinfo{number}{2}
  (\bibinfo{year}{2018}), \bibinfo{pages}{16}.
\newblock


\bibitem[\protect\citeauthoryear{Pang, Jia, Zhang, Zhang, Huang, and Yin}{Pang
  et~al\mbox{.}}{2015}]%
        {pang2015unsupervised}
\bibfield{author}{\bibinfo{person}{Junbiao Pang}, \bibinfo{person}{Fei Jia},
  \bibinfo{person}{Chunjie Zhang}, \bibinfo{person}{Weigang Zhang},
  \bibinfo{person}{Qingming Huang}, {and} \bibinfo{person}{Baocai Yin}.}
  \bibinfo{year}{2015}\natexlab{}.
\newblock \showarticletitle{Unsupervised web topic detection using a ranked
  clustering-like pattern across similarity cascades}.
\newblock \bibinfo{journal}{\emph{IEEE Transactions on Multimedia}}
  \bibinfo{volume}{17}, \bibinfo{number}{6} (\bibinfo{year}{2015}),
  \bibinfo{pages}{843--853}.
\newblock


\bibitem[\protect\citeauthoryear{Pennington, Socher, and Manning}{Pennington
  et~al\mbox{.}}{2014}]%
        {pennington2014glove}
\bibfield{author}{\bibinfo{person}{Jeffrey Pennington},
  \bibinfo{person}{Richard Socher}, {and} \bibinfo{person}{Christopher~D
  Manning}.} \bibinfo{year}{2014}\natexlab{}.
\newblock \showarticletitle{Glove: Global Vectors for Word Representation}. In
  \bibinfo{booktitle}{\emph{Proceedings of the 2014 Conference on Empirical
  Methods on Natural Language Processing}}, Vol.~\bibinfo{volume}{14}.
  \bibinfo{pages}{1532--43}.
\newblock


\bibitem[\protect\citeauthoryear{Raamkumar, Foo, and Pang}{Raamkumar
  et~al\mbox{.}}{2017}]%
        {raamkumar2017using}
\bibfield{author}{\bibinfo{person}{Aravind~Sesagiri Raamkumar},
  \bibinfo{person}{Schubert Foo}, {and} \bibinfo{person}{Natalie Pang}.}
  \bibinfo{year}{2017}\natexlab{}.
\newblock \showarticletitle{Using author-specified keywords in building an
  initial reading list of research papers in scientific paper retrieval and
  recommender systems}.
\newblock \bibinfo{journal}{\emph{Information Processing \& Management}}
  \bibinfo{volume}{53}, \bibinfo{number}{3} (\bibinfo{year}{2017}),
  \bibinfo{pages}{577--594}.
\newblock


\bibitem[\protect\citeauthoryear{Ramaboa and Fish}{Ramaboa and Fish}{2018}]%
        {ramaboa2018keyword}
\bibfield{author}{\bibinfo{person}{Kutlwano~KKM Ramaboa} {and}
  \bibinfo{person}{Peter Fish}.} \bibinfo{year}{2018}\natexlab{}.
\newblock \showarticletitle{Keyword length and matching options as indicators
  of search intent in sponsored search}.
\newblock \bibinfo{journal}{\emph{Information Processing \& Management}}
  \bibinfo{volume}{54}, \bibinfo{number}{2} (\bibinfo{year}{2018}),
  \bibinfo{pages}{175--183}.
\newblock


\bibitem[\protect\citeauthoryear{Ramezani, Khodadadi, and Rabiee}{Ramezani
  et~al\mbox{.}}{2018}]%
        {ramezani2018community}
\bibfield{author}{\bibinfo{person}{Maryam Ramezani}, \bibinfo{person}{Ali
  Khodadadi}, {and} \bibinfo{person}{Hamid~R Rabiee}.}
  \bibinfo{year}{2018}\natexlab{}.
\newblock \showarticletitle{Community Detection Using Diffusion Information}.
\newblock \bibinfo{journal}{\emph{ACM Transactions on Knowledge Discovery from
  Data (TKDD)}} \bibinfo{volume}{12}, \bibinfo{number}{2}
  (\bibinfo{year}{2018}), \bibinfo{pages}{20}.
\newblock


\bibitem[\protect\citeauthoryear{Siersdorfer, San~Pedro, and
  Sanderson}{Siersdorfer et~al\mbox{.}}{2009}]%
        {siersdorfer2009automatic}
\bibfield{author}{\bibinfo{person}{Stefan Siersdorfer}, \bibinfo{person}{Jose
  San~Pedro}, {and} \bibinfo{person}{Mark Sanderson}.}
  \bibinfo{year}{2009}\natexlab{}.
\newblock \showarticletitle{Automatic video tagging using content redundancy}.
  In \bibinfo{booktitle}{\emph{Proceedings of the 32nd International ACM SIGIR
  Conference on Research and Development in Information Retrieval}}. ACM,
  \bibinfo{pages}{395--402}.
\newblock


\bibitem[\protect\citeauthoryear{Socher, Lin, Manning, and Ng}{Socher
  et~al\mbox{.}}{2011}]%
        {socher2011parsing}
\bibfield{author}{\bibinfo{person}{Richard Socher}, \bibinfo{person}{Cliff~C
  Lin}, \bibinfo{person}{Chris Manning}, {and} \bibinfo{person}{Andrew~Y Ng}.}
  \bibinfo{year}{2011}\natexlab{}.
\newblock \showarticletitle{Parsing natural scenes and natural language with
  recursive neural networks}. In \bibinfo{booktitle}{\emph{Proceedings of the
  28th International Conference on Machine Learning (ICML-11)}}.
  \bibinfo{pages}{129--136}.
\newblock


\bibitem[\protect\citeauthoryear{Tarjan}{Tarjan}{1975}]%
        {tarjan1975efficiency}
\bibfield{author}{\bibinfo{person}{Robert~Endre Tarjan}.}
  \bibinfo{year}{1975}\natexlab{}.
\newblock \showarticletitle{Efficiency of a good but not linear set union
  algorithm}.
\newblock \bibinfo{journal}{\emph{Journal of the ACM (JACM)}}
  \bibinfo{volume}{22}, \bibinfo{number}{2} (\bibinfo{year}{1975}),
  \bibinfo{pages}{215--225}.
\newblock


\bibitem[\protect\citeauthoryear{Tarjan}{Tarjan}{1979}]%
        {tarjan1979class}
\bibfield{author}{\bibinfo{person}{Robert~Endre Tarjan}.}
  \bibinfo{year}{1979}\natexlab{}.
\newblock \showarticletitle{A class of algorithms which require nonlinear time
  to maintain disjoint sets}.
\newblock \bibinfo{journal}{\emph{Journal of computer and system sciences}}
  \bibinfo{volume}{18}, \bibinfo{number}{2} (\bibinfo{year}{1979}),
  \bibinfo{pages}{110--127}.
\newblock


\bibitem[\protect\citeauthoryear{Wang, Song, Roth, Zhang, and Han}{Wang
  et~al\mbox{.}}{2016b}]%
        {wang2016world}
\bibfield{author}{\bibinfo{person}{Chenguang Wang}, \bibinfo{person}{Yangqiu
  Song}, \bibinfo{person}{Dan Roth}, \bibinfo{person}{Ming Zhang}, {and}
  \bibinfo{person}{Jiawei Han}.} \bibinfo{year}{2016}\natexlab{b}.
\newblock \showarticletitle{World knowledge as indirect supervision for
  document clustering}.
\newblock \bibinfo{journal}{\emph{ACM Transactions on Knowledge Discovery from
  Data (TKDD)}} \bibinfo{volume}{11}, \bibinfo{number}{2}
  (\bibinfo{year}{2016}), \bibinfo{pages}{13}.
\newblock


\bibitem[\protect\citeauthoryear{Wang, Gu, Guo, Chen, Leung, and Sun}{Wang
  et~al\mbox{.}}{2017}]%
        {wang2017crowdsourcing}
\bibfield{author}{\bibinfo{person}{Kun Wang}, \bibinfo{person}{Liqiu Gu},
  \bibinfo{person}{Song Guo}, \bibinfo{person}{Hongbin Chen},
  \bibinfo{person}{Victor~CM Leung}, {and} \bibinfo{person}{Yanfei Sun}.}
  \bibinfo{year}{2017}\natexlab{}.
\newblock \showarticletitle{Crowdsourcing-based content-centric network: a
  social perspective}.
\newblock \bibinfo{journal}{\emph{IEEE Network}} \bibinfo{volume}{31},
  \bibinfo{number}{5} (\bibinfo{year}{2017}), \bibinfo{pages}{28--34}.
\newblock


\bibitem[\protect\citeauthoryear{Wang, Qi, Shu, Deng, and Rodrigues}{Wang
  et~al\mbox{.}}{2016a}]%
        {wang2016toward}
\bibfield{author}{\bibinfo{person}{Kun Wang}, \bibinfo{person}{Xin Qi},
  \bibinfo{person}{Lei Shu}, \bibinfo{person}{Der-jiunn Deng}, {and}
  \bibinfo{person}{Joel~JPC Rodrigues}.} \bibinfo{year}{2016}\natexlab{a}.
\newblock \showarticletitle{Toward trustworthy crowdsourcing in the social
  internet of things}.
\newblock \bibinfo{journal}{\emph{IEEE Wireless Communications}}
  \bibinfo{volume}{23}, \bibinfo{number}{5} (\bibinfo{year}{2016}),
  \bibinfo{pages}{30--36}.
\newblock


\bibitem[\protect\citeauthoryear{Wu, Yang, and He}{Wu et~al\mbox{.}}{2012}]%
        {wu2012chinese}
\bibfield{author}{\bibinfo{person}{Benbin Wu}, \bibinfo{person}{Jing Yang},
  {and} \bibinfo{person}{Liang He}.} \bibinfo{year}{2012}\natexlab{}.
\newblock \showarticletitle{Chinese hownet-based multi-factor word similarity
  algorithm integrated of result modification}. In
  \bibinfo{booktitle}{\emph{International Conference on Neural Information
  Processing}}. Springer, \bibinfo{pages}{256--266}.
\newblock


\bibitem[\protect\citeauthoryear{Wu, Zhong, Tan, Horner, and Yang}{Wu
  et~al\mbox{.}}{2014}]%
        {wu2014crowdsourced}
\bibfield{author}{\bibinfo{person}{Bin Wu}, \bibinfo{person}{Erheng Zhong},
  \bibinfo{person}{Ben Tan}, \bibinfo{person}{Andrew Horner}, {and}
  \bibinfo{person}{Qiang Yang}.} \bibinfo{year}{2014}\natexlab{}.
\newblock \showarticletitle{Crowdsourced time-sync video tagging using temporal
  and personalized topic modeling}. In \bibinfo{booktitle}{\emph{Proceedings of
  the 20th ACM SIGKDD International Conference on Knowledge Discovery and Data
  Mining}}. ACM, \bibinfo{pages}{721--730}.
\newblock


\bibitem[\protect\citeauthoryear{Xu and Zhang}{Xu and Zhang}{2017}]%
        {xu2017bridging}
\bibfield{author}{\bibinfo{person}{Linli Xu} {and} \bibinfo{person}{Chao
  Zhang}.} \bibinfo{year}{2017}\natexlab{}.
\newblock \showarticletitle{Bridging Video Content and Comments: Synchronized
  Video Description with Temporal Summarization of Crowdsourced Time-Sync
  Comments.}. In \bibinfo{booktitle}{\emph{AAAI}}. \bibinfo{pages}{1611--1617}.
\newblock


\bibitem[\protect\citeauthoryear{Yan, Guo, Lan, and Cheng}{Yan
  et~al\mbox{.}}{2013}]%
        {yan2013biterm}
\bibfield{author}{\bibinfo{person}{Xiaohui Yan}, \bibinfo{person}{Jiafeng Guo},
  \bibinfo{person}{Yanyan Lan}, {and} \bibinfo{person}{Xueqi Cheng}.}
  \bibinfo{year}{2013}\natexlab{}.
\newblock \showarticletitle{A biterm topic model for short texts}. In
  \bibinfo{booktitle}{\emph{Proceedings of the 22nd international conference on
  World Wide Web}}. ACM, \bibinfo{pages}{1445--1456}.
\newblock


\bibitem[\protect\citeauthoryear{Yang, Ruan, Gao, Wang, Ran, and Jia}{Yang
  et~al\mbox{.}}{2017}]%
        {yang2017crowdsourced}
\bibfield{author}{\bibinfo{person}{Wenmian Yang}, \bibinfo{person}{Na Ruan},
  \bibinfo{person}{Wenyuan Gao}, \bibinfo{person}{Kun Wang},
  \bibinfo{person}{Wensheng Ran}, {and} \bibinfo{person}{Weijia Jia}.}
  \bibinfo{year}{2017}\natexlab{}.
\newblock \showarticletitle{Crowdsourced time-sync video tagging using semantic
  association graph}. In \bibinfo{booktitle}{\emph{Multimedia and Expo (ICME),
  2017 IEEE International Conference on}}. IEEE, \bibinfo{pages}{547--552}.
\newblock


\bibitem[\protect\citeauthoryear{Yin and Wang}{Yin and Wang}{2014}]%
        {yin2014dirichlet}
\bibfield{author}{\bibinfo{person}{Jianhua Yin} {and} \bibinfo{person}{Jianyong
  Wang}.} \bibinfo{year}{2014}\natexlab{}.
\newblock \showarticletitle{A dirichlet multinomial mixture model-based
  approach for short text clustering}. In \bibinfo{booktitle}{\emph{Proceedings
  of the 20th ACM SIGKDD international conference on Knowledge discovery and
  data mining}}. ACM, \bibinfo{pages}{233--242}.
\newblock


\bibitem[\protect\citeauthoryear{Yin and Wang}{Yin and Wang}{2016}]%
        {yin2016model}
\bibfield{author}{\bibinfo{person}{Jianhua Yin} {and} \bibinfo{person}{Jianyong
  Wang}.} \bibinfo{year}{2016}\natexlab{}.
\newblock \showarticletitle{A model-based approach for text clustering with
  outlier detection}. In \bibinfo{booktitle}{\emph{Proceedings of Data
  Engineering (ICDE), 2016 IEEE 32nd International Conference on}}. IEEE,
  \bibinfo{pages}{625--636}.
\newblock


\bibitem[\protect\citeauthoryear{Yu, Wang, He, Tian, Lu, and Guo}{Yu
  et~al\mbox{.}}{2015}]%
        {yu2015discovering}
\bibfield{author}{\bibinfo{person}{Zhiwen Yu}, \bibinfo{person}{Zhu Wang},
  \bibinfo{person}{Huilei He}, \bibinfo{person}{Jilei Tian},
  \bibinfo{person}{Xinjiang Lu}, {and} \bibinfo{person}{Bin Guo}.}
  \bibinfo{year}{2015}\natexlab{}.
\newblock \showarticletitle{Discovering information propagation patterns in
  microblogging services}.
\newblock \bibinfo{journal}{\emph{ACM Transactions on Knowledge Discovery from
  Data (TKDD)}} \bibinfo{volume}{10}, \bibinfo{number}{1}
  (\bibinfo{year}{2015}), \bibinfo{pages}{7}.
\newblock


\bibitem[\protect\citeauthoryear{Zhu}{Zhu}{2004}]%
        {zhu2004recall}
\bibfield{author}{\bibinfo{person}{Mu Zhu}.} \bibinfo{year}{2004}\natexlab{}.
\newblock \showarticletitle{Recall, precision and average precision}.
\newblock \bibinfo{journal}{\emph{Department of Statistics and Actuarial
  Science, University of Waterloo, Waterloo}}  \bibinfo{volume}{2}
  (\bibinfo{year}{2004}), \bibinfo{pages}{30}.
\newblock


\end{thebibliography}

\end{document}